\documentclass[numsec,webpdf,traditional,medium]{oup-authoring-template} 
\usepackage{booktabs}
\usepackage{pifont}
\usepackage{float}
\usepackage{prodint}

\onecolumn 

\graphicspath{{Fig/}}

\theoremstyle{thmstyleone}%
\newtheorem{theorem}{Theorem}
\newtheorem{proposition}[theorem]{Proposition}%
\theoremstyle{thmstyletwo}%
\theoremstyle{thmstylethree}%

\def\Var{\mbox{\rm Var}}

\def\Pr{\mbox{\rm Pr}}

\newcommand{\cG}{\mathcal{G}}

\newcommand{\cA}{\mathcal{A}}
\newcommand{\cT}{\mathcal{T}}
\newcommand{\cF}{\mathcal{F}}
\newcommand{\cW}{\mathcal{W}}
\newcommand{\cH}{\mathcal{H}}
\newcommand{\cI}{\mathcal{I}}
\newcommand{\cD}{\mathcal{D}}

\newcommand{\cK}{\mathcal{K}}

\newcommand{\cB}{\mathcal{B}}

\newcommand{\cY}{\mathcal{Y}}

\newcommand{\sumi}{\sum_{i=1}^n}

\newcommand{\sumj}{\sum_{j=1}^n}

\newcommand{\real}{\mathbb{R}}
\newcommand{\cp}{\stackrel{p}{\rightarrow}}

\newcommand{\omitt}[1]{{}}
\newcommand{\pss}[1]{{(#1)}}

\newcommand{\sttwo}{S_{T_2}}
\newcommand{\sttw}{S_{T_2|T_1}}
\newcommand{\cKbc}{\mathcal{K}_{bc}}
\newcommand{\cU}{\mathcal{U}}
\newcommand{\Op}{O_P}
\newcommand{\op}{o_P}
\newcommand{\frn}{\frac{1}{n}}
\newcommand{\frnh}{\frac{1}{nh}}
\newcommand{\wh}[1]{\widehat{#1}}
\newcommand{\sttwt}{\widetilde{S}_{T_2|T_1}}
\newcommand{\summ}{\sum_{m=1}^n}
\newcommand{\sttwoh}{\widehat{S}_{T_2}}
\newcommand{\sttwh}{\widehat{S}_{T_2|T_1}}

\begin{document}

\journaltitle{Submitted}
\DOI{DOI HERE}
\copyrightyear{2022}
\pubyear{2019}
\access{Advance Access Publication Date: Day Month Year}
\appnotes{Paper}

\firstpage{1}


\title[Cumulative Incidence Function Estimation]{Efficient Cumulative Incidence Estimation in Biobank Studies Using All Prevalent and Incident Events}

\author[1,$\ast$]{David M. Zucker\ORCID{0000-0002-2104-4913}}
\author[2]{Malka Gorfine\ORCID{0000-0002-1577-6624}}

\authormark{Zucker and Gorfine}

\address[1]{\orgdiv{Department of Statistics and Data Science}, \orgname{Hebrew University of Jerusalem}, 
\orgaddress{\street{Mount Scopus}, \postcode{91905} \state{Jerusalem}, \country{Israel}}}
\address[2]{\orgdiv{Department of Statistics and Operations Research}, \orgname{Tel Aviv University},
\orgaddress{\street{}, \postcode{}, \state{Ramat Gan}, \country{Israel}}}

\corresp[$\ast$]{Corresponding author: \href{email:email-id.com}{david.zucker@mail.huji.ac.il}}

\received{Date}{0}{Year}
\revised{Date}{0}{Year}
\accepted{Date}{0}{Year}



\abstract{
Population-based biobanks, now established in many countries, offer opportunities for large-scale studies investigating
the incidence of various diseases. Biobank data is typically collected from a study cohort recruited over a defined calendar
period, with subjects entering the study at various ages falling between $R_L$ and $R_U$. This work focuses on biobank data 
that includes individuals in whom onset of the disease of interest occurred before recruitment, termed prevalent cases, along 
with individuals initially recruited as disease-free in whom disease onset occurred during the follow-up period. We propose a 
novel cumulative incidence function (CIF) estimator that goes beyond existing methods in that it incorporates all disease cases,
both prevalent and incident, irrespective of their subsequent life course. In particular, the new method can handle situations 
involving diseases that can occur at young ages with long survival after disease onset. Asymptotic properties of
the new method are established and a simulation study is presented examining the performance of the method. We illustrate the use of the method and highlight its advantages over existing
methods with an application to cancer data from the UK biobank.}

\keywords{Survival analysis, left truncation, illness-death model, Aalen-Johansen estimator, prevalent cases}


\maketitle

\section{Introduction}

\subsection{Setup and Related Works}
We consider the problem of estimating the incidence of some disease using population-based biobank data.
We work with  the illness-death model, also known as the semi-competing risks model.
Here, individuals start in a healthy state and can then move to the diseased state (transition $0 \rightarrow 1$), signifying being diagnosed with the disease of interest, and from there to the dead state ($1 \rightarrow 2$). Individuals can also move directly from the healthy state to the dead state ($0 \rightarrow 2$) without experiencing the disease of interest. 
We work on the time scale of age, with the origin being birth.
Let the random variables $T_1$ and $T_2$ represent the age at diagnosis and age at death, respectively. 
For those who died without having the disease, we set $T_1 = \infty$. Our goal is to use biobank data to estimate the probability of having the disease by time $t$, which is the cumulative incidence function (CIF), $G_1(\cdot)$,
given by
$$
G_1(t) = \Pr (T_1 \leq t, T_2 > T_1), \quad t \in [0,\tau],
$$
for a constant $\tau > 0$ representing the maximum age at end of follow-up. In the case where the study design recruits only subjects aged between $[R_L, R_U]$, it is not possible to provide a nonparametric estimator of the CIF for individuals who died before age $R_L > 0$, and we therefore work with the alternative goal of estimating the conditional CIF $G_1(\cdot|T_2 > R_L)$, defined as
\begin{equation}
G_1(t|T_2>R_L) = \Pr (T_1 \leq t, T_2 > T_1|T_2>R_L) \, \,\,\, t \in [0,\tau].
\label{cifwun}
\end{equation}

Biobanks have quietly changed what is feasible in observational medicine. Instead of assembling a disease-specific cohort and following it for years, investigators can now start from a population resource that already links baseline assessments to registries of diagnoses and deaths, often at national scale. This creates unprecedented power for studying time-to-event outcomes, especially rare diseases and long-latency endpoints, and for comparing risks across many phenotypes using a common data backbone. The price of this convenience is that the data-generating mechanism is no longer the classical ``start follow-up at time 0 and watch events unfold." Biobank participation is itself conditioned on survival and eligibility, and those design features reshape what can be estimated and how. The present paper targets a key challenge arising under this design, which, surprisingly, has yet to be addressed in an efficient manner:
estimation of the CIF when diagnosis is non-terminal and death may occur before or after diagnosis, while recruitment occurs only within an age window.

Most large biobanks recruit participants over a fixed calendar period, but only if their age at enrollment falls in a pre-specified interval $[R_L,R_U]$. In the UK Biobank, for example, $R_L=40$ and $R_U=69$. Follow-up then proceeds prospectively through linked electronic health records and registries. Two implications are immediate. First, by design, we never observe individuals who die before $R_L$, so any incidence statement must be interpreted conditionally on survival to $R_L$. Second, even among those who do survive past $R_L$, individuals enter at different ages $R$, which induces left truncation (delayed entry): people are only observed if they alive to their enrollment age. These features are well known in principle, but in practice they become technically difficult
to handle 
\citep{keiding1991age,saarela2009joint,vakulenko2016comparing,gzs2025}.

These issues are especially acute in the semi-competing risks setting when the target is the CIF of a given disease, conditional on survival to $R_L$, namely $G_1(t\mid T_2>R_L)$. In biobank data, individuals naturally fall into four observable categories: (i) prevalent cases, diagnosed before recruitment; (ii) incident cases, diagnosed during follow-up; (iii) individuals who die without a recorded diagnosis of the disease; and (iv) right-censored observations who remain alive and disease-free by the end of follow-up. Prevalent cases can represent a substantial fraction of all observed cases, even for relatively rare diseases, and discarding them can lead to loss of information and, more importantly, to estimating a different estimand than $G_1(t\mid T_2>R_L)$,
as explained below. To illustrate, among 243,871 UK Biobank female participants, breast cancer has approximately 17,000 recorded cases, of which approximately 50\% are prevalent at recruitment. This paper addresses the resulting methodological gap by developing a CIF estimator that incorporates all disease cases, prevalent and incident, thereby aligning estimation with the conditional biobank estimand while retaining efficiency in settings where post-diagnosis survival may be long. 

Estimating from left-truncated and right-censored data is typically done either through risk-set correction or 
inverse-probability weighting (IPW). The well-known Aalen-Johansen (AJ) estimator \citep{aalen1978empirical} can be easily 
adjusted for left-truncation (see \cite{allignol2010note} and Section 2.1 below) using a risk-set correction. An individual 
enters the risk set after its left-truncation (i.e., enrollment) time, and the individual stays in the risk set until its event 
or censoring time, whichever comes first. However, this risk-set adjustment omits prevalent cases from the estimation procedure, 
since the event time precedes the entry time. Besides the potential reduction in efficiency, it also implies that the AJ 
estimator can estimate the CIF only for $t > R_L$, conditional on being disease-free at the time of recruitment $(T_1 > R)$. 
That is, the target of the AJ estimator is not $G_1(t|T_2>R_L)$ as defined in (\ref{cifwun}), but rather
\begin{equation}
	G_1(t|T_1>R_L, T_2>R_L) = \Pr (T_1 \leq t, T_2 > T_1|T_1>R_L, T_2>R_L) \, \,\,\, t \in [0,\tau].	
	\label{ciftwo}
\end{equation}
The functions $G_1(t|T_2>R_L)$
and $G_1(t|T_1 > R_L, T_2 > R_L)$ coincide in the case where the probability of disease onset before $R_L$ is zero,
and approximately coincide when this probability is negligible, but otherwise they differ.

\cite{chang2006nonparametric} and \cite{vakulenko2017nonparametric} provided nonparametric estimators employing IPW methods. Unlike the Aalen-Johansen estimator, these approaches do not exclude prevalent observations, but instead correct for
the sampling bias by IPW. Their methods assign positive
masses only to completely uncensored observations, and the weights
depend on the distributions of censoring and truncation. 
\cite{gzs2025} (hereafter GZS) present an alternative IPW estimator.
As explained in GZS, the method of \citeauthor{chang2006nonparametric} and the method of
\cite{vakulenko2017nonparametric}
deliver guaranteed consistent estimators only under a strong assumption on the censoring distribution,
and if this assumption is violated, the estimators can be very seriously biased. 
The method of GZS delivers good performance under a broader range of scenarios.

However, the method of GZS suffers from one major limitation: It does not cover scenarios with long survival 
after disease onset, such that the proportion of diseased individuals who died during the follow-up period
is low. This is because the method of GZS, like those of \cite{chang2006nonparametric} and 
\cite{vakulenko2017nonparametric}, assigns positive masses only to completely uncensored observations.

The aim of the present paper is to present a new CIF estimator that does not suffer from this limitation.
The new estimator assigns positive mass to all disease cases, irrespective of whether the case is a
prevalent case or an incident case, and irrespective of whether the individual died during the
study or survived to the end of follow-up. Thus, this is the first paper to present a 
fully nonparametric and consistent.
CIF estimator in the illness-death model that incorporates all disease cases.

The practical importance of this advance is illustrated in our analysis of breast cancer in the UK Biobank. Breast cancer is a setting in which survival after diagnosis is often long, and consequently methods that assign positive mass only to completely observed disease–death trajectories may perform poorly. Indeed, as shown in our simulation study in
Section~5, the GZS estimator is substantially biased in this setting. 
In our analysis of the UK Biobank data, the left-truncation adjusted AJ estimator and the proposed estimator produced highly similar estimates for the breast cancer CIF while
the GZS estimate was markedly lower. In addition,
the proposed estimator is considerably more efficient than the AJ estimator because it makes use of all observed breast cancer cases, including prevalent cases and cases not followed until death. In particular, the ratio of the average widths of the 95\% simultaneous confidence bands for the AJ and proposed estimators is 0.0068/0.0035=1.943, nearly a two-fold efficiency gain in favor of the proposed method.

\subsection{Contributions}
Our contributions can be summarized as follows:
\begin{enumerate}
	\item  \,
	We introduce a novel nonparametric CIF estimator that efficiently utilizes all disease cases. 
	Incorporating all disease cases offers three advantages: (1) effective CIF estimation in broad range
	of scenarios, including scenarios in which the proportion of diseased individuals who died during the follow-up
	period is low, (2) substantially enhanced efficiency in many scenarios, and (3) enabling the estimation of CIF at ages preceding $R_L$. 
	\item  \,
	We establish the consistency and asymptotic normality of the proposed estimator. Additionally, we provide pointwise confidence intervals and simultaneous confidence bands for the CIF.
	\item  \,
	We present an extensive simulation study and an analysis of data from the UKB on
	multiple types of cancer to demonstrate the performance and practical utility of the proposed estimator. An R package \texttt{illdthCIF1} implementing the proposed estimator is provided.
\end{enumerate}

\section{CIF Estimators}

\subsection{Preliminaries}

Consider $n$ independent observations.
Let $T_{1i}$ and $T_{2i}$ be the age at disease diagnosis and
age at death, respectively, of the $i$th observation, $i = 1,\ldots,n$. 
If the participant dies without being diagnosed with the disease under study, we set $T_{1i} = \infty$.
We assume that the pair $(T_1,T_2)$ has an absolutely continuous distribution, aside for the atom at $\infty$ 
in the distribution of $T_1$ corresponding to individuals in whom the disease of interest did not occur. We denote by $t_{1min}$ the minimum of the support of $T_1$.
Define $C_i$ as the right-censoring time and $R_i$ as the age at recruitment, where $R_L \leq R_i \leq R_U$, signifying that 
subjects are recruited at ages between $R_L$ and $R_U$. 
All densities, survival functions, and probability calculations in this paper are conditional on $T_2 > R_L$.
We assume that $R_i$ is independent of $(T_{1i}, T_{2i})$ and that $C_i=R_i+W_i$ with $W_i$ independent of $(T_{1i}, T_{2i}, R_i)$.
Define $V_{1i} = \min(T_{1i}, T_{2i}, C_i)$, 
$V_{2i}=\min(T_{2i},C_i)$, and $\delta_{2i} = I(T_{2i} \leq C_{i})$. 
Further, define $Z_i$ to be equal to 0 if individual $i$ was censored at time $V_{1i}$,
1 if individual $i$ experienced disease onset at at time $V_{1i}$, and 2 if
individual $i$ died without disease at time $V_{1i}$, and let $\delta_{1i}=I(Z_i=1)$.
The observed data consist of $\{V_{1i},V_{2i}, Z_i, \delta_{2i}, R_i \, , \, 
i=1,\ldots,n\}$. To be recruited, a participant must be alive at the recruitment time, i.e., $T_{2i} \geq R_i$.

Under the above assumptions, the conditional joint density of $(T_{2i}, T_{2i}, R_i, W_i)$ given $R_i \leq T_{2i}$ is
\begin{equation}
	f_{T_{1},T_{2},R,W|R \leq T_{2}}(t_1,t_2,r,w) = a^{-1} f_{T_{1},T_{2}}(t_1,t_2) f_{R}(r) f_{W}(w)  I(t_1 \leq t_2)I (r \leq~t_2), 
	\label{jden}
\end{equation}
where $a$ is the normalizing constant. Note that
\begin{align*}
	f_{R|R \leq T_{2}}(r) & = a^{-1} \int f_{T_2}(t_2) f_{R}(r) I(r \leq t_2) \, dt_2 =
	a^{-1} f_{R}(r)  S_{T_2}(r).
\end{align*}
Therefore
\begin{equation}
	f_{R}(r) = a f_{R|R \leq T_{2}}(r)S_{T_2}(r)^{-1}.
	\label{rr}
\end{equation}
We will make use of (\ref{rr}) later.
Quantities of the form
\begin{equation}
	\mathcal{I} = \int \phi(r) f_{R|R \leq T_{2}}(r) \, dr 
	\label{eye}
\end{equation}
can be estimated using the empirical estimator
\begin{equation}
	\widehat{\mathcal{I}} = \frac{1}{n} \sum_{j=1}^n \phi(R_j).
	\label{eest}
\end{equation}
The survival function $\sttwo$ can be estimated using the standard Kaplan-Meier estimator $\widehat{S}_{T_2}$ adjusted for left-truncation, as described, for example, by \cite{tsai1987}.

\subsection{Left-Truncation Adjusted Aalen-Johansen Estimator}

The classical estimator of Aalen and Johansen (1978) is constructed as follows.
Define $\tilde{T}=T_1 \wedge T_2$, 
let $N_{1i}(u) = \delta_{1i} I(V_{1i} \leq u)$ be the counting process for disease occurrence, and let
$Y_{1i}(u)=I(V_{1i} \geq u)$ be the at-risk indicator, equal to 1 if the participant is still at risk 
and equal to 0 if not.
The AJ estimator of the CIF is then given by
$$
\widehat{G}^{AJ}_1(t_1^*)= \int_0^{t_1^*} \widehat{S}_{T^*}(u-) \frac{\sumi dN_{1i}(u)}{\sumi Y_{1i}(u)} \, ,
$$
where $\widehat{S}_{\tilde{T}}(\cdot)$ is the Kaplan-Meier estimator of the survival function of $\tilde{T}$. For data with left truncation, estimation 
is based on risk-set adjustment for delayed entry. Namely, at time $t$, the at-risk process of individual $i$ is defined by $Y_{1i}(t) =I(R_i \leq t \leq V_{1i})$. 
This estimator excludes prevalent cases, i.e., individuals with $V_{1i} < R_i$. 
Therefore, the estimand of $\widehat{G}^{AJ}_1(t)$ is not $G_1(t|T_2 > R_L)$ as defined by
(\ref{cifwun}) but rather $G_1(t|T_1 > R_L, T_2 > R_L)$ as defined by (\ref{ciftwo}).
The estimator $\widehat{G}^{AJ}_1(t)$ is a consistent estimator of this estimand.

\subsection{The Estimator of Gorfine, Zucker, and Shoham}

\cite{gzs2025} proposed the estimator
\begin{equation}
	\widehat{G}_1^{GZS}(t_1^*) = \frac{1}{n} \sumi \delta_{1i} \delta_{2i} I(V_{1i} \leq t_1^*) 
	\left(\frac{\widehat{S}_{T_2}(V_{2i})}{n^{-1} \sum_{j=1}^n Y_{2i}(V_{2i})}\right),
	\label{gzs}
\end{equation}
where $Y_{2j}(u) = I(R_i < u \leq C_i)$. They proved consistency and provided an asymptotic representation of a truncated version of
$\widehat{G}_1^{GZS}(t)$ and conjectured that these results extend to the original estimator. 
In addition, they presented simulation
results for the original estimator demonstrating good performance with lower variance than the AJ
estimator in many scenarios. However, as noted previously,
it cannot handle scenarios where the proportion of diseased individuals who died during the follow-up period is low.

\subsection{The Proposed Estimator}
Our proposed estimator, of a form similar to the GZS estimator, is
\begin{equation}
	\widehat{G}_1(t_1^*) = \frac{1}{n} \sumi \delta_{1i} I(V_{1i} \leq t_1^*) \widehat{Q}(V_{1i}),
	\label{est}	
\end{equation}
where $\widehat{Q}$ is an estimator of a quantity $Q$ that is chosen so as to yield a consistent estimator of the CIF, as described below.
For any function $Q$, starting from (\ref{jden}), we have
\begin{align*}
	& E[\delta_{1i}I(V_{1i}\leq t_1^*)Q(V_{1i})] \\
	& \hspace*{5mm} = a^{-1} \int \int \int \int f_{T_{1},T_{2}}(t_1,t_2) f_{R}(r)f_{W}(w) \\
	& \hspace*{30mm} Q(t_1) I(t_2 \geq r )I(t_1 \leq t_1^*) I(r+w \geq t_1) \, dt_1 \, dt_2 \, dr \, dw \\
	& \hspace*{5mm} = a^{-1} \int \int \int \int f_{T_{1}}(t_1) f_{T_{2}|T_1}(t_2|t_1) f_{R}(r)f_{W}(w) \\
	& \hspace*{30mm} Q(t_1) I(t_2 \geq r )I(t_1 \leq t_1^*) I(r+w \geq t_1) \, dt_1 \, dt_2 \, dr \, dw \\
	& \hspace*{5mm} = a^{-1} \int f_{T_{1}}(t_1) I(t_1 \leq t_1^*) Q(t_1) \int \int f_{R}(r)f_{W}(w) I(r+w \geq t_1) \\
	& \hspace*{30mm} \left\{\int  f_{T_{2}|T_1}(t_2|t_1) I(t_2 \geq r) \, dt_2 \right\} dt_1 \, dr \, dw \\
	& \hspace*{5mm} = a^{-1} \int f_{T_{1}}(t_1) I(t_1 \leq t_1^*) Q(t_1) 
	\int \int f_{R}(r)f_{W}(w) S_{T_2|T_1}(r|t_1) I(r+w \geq t_1) \, dt_1 \, dr \, dw \\
	& \hspace*{5mm} = a^{-1} \int f_{T_{1}}(t_1) I(t_1 \leq t_1^*) Q(t_1) 
	\int f_{R}(r)f_{W}(w) S_{T_2|T_1}(r|t_1) S_W((t_1-r)_+) \, dt_1 \, dr  \\
	& \hspace*{5mm} = \int f_{T_{1}}(t_1) I(t_1 \leq t_1^*) Q(t_1) 
	\int \left\{ f_{R|R\leq T_2}(r) S_{T_2}(r)^{-1} \right\} S_{T_2|T_1}(r|t_1) S_W((t_1-r)_+) \, dr \, dt_1, 
\end{align*}
where $z_+ = \max(z,0)$ and in the last step we have used (\ref{rr}).
Thus, defining
\begin{align*}
B(t_1) & = \int_{R_L}^{R_U}  f_{R|R\leq T_2}(r) S_{T_2}(r)^{-1} 
S_{T_2|T_1}(r|t_1) S_W((t_1-r)_+) \, dr,
\end{align*}
if we take $Q(t_1) = B(t_1)^{-1}$, we get
\begin{align*}
E[\delta_{1i}I(V_{1i}\leq t_1^*)Q(V_{1i})] 
= P(T_1 \leq t_1^*, T_2>T_1|T_2>R_L)
= G_1(t_1^*).
\end{align*}
We assume that $B(t_1)$ is bounded below.
Following (\ref{eest}), we can estimate $B(V_{1i})$ by
\begin{align}
	\widehat{B}(V_{1i}) = \frac{1}{n} \sum_{j=1}^n \widehat{S}_{T_2}(R_j)^{-1} \widehat{S}_{T_2|T_1}(R_j|V_{1i}) 
	\widehat{S}_{W}((V_{1i}-R_j)_+).
	\label{qhat}
\end{align}
Here, $\widehat{S}_{T_2}$ can be obtained using the Kaplan-Meier estimate, as noted previously, and, similarly, $\widehat{S}_W$ can be
estimated using the Kaplan-Meier estimator applied to the time from recruitment to censoring.
For estimation of $S_{T_2|T_1}$, we use the Nelson-Aalen version of
a well-known kernel-type nonparametric survival regression estimator initially proposed by
\cite{beran1981} and later studied by in \cite{dab1987}, \citet[Proposition 4.3]{vankeil1999}, and other works.
In the present setting, the estimator of the conditional cumulative hazard function
$\Lambda_{T_2|T_1}(t_2|t_1)$ is given by
\begin{equation}
	\widehat{\Lambda}_{T_2|T_1}(t_2|t_1) = \int_{t_1}^{t_2} \frac{(nh_n)^{-1} \sum_{m=1}^n \delta_{1m} \cK((V_{1m}-t_1)/h_n) dN_{2m}(u)}
	{(nh_n)^{-1} \sum_{m=1}^n \delta_{1m} \cK((V_{1m}-t_1)/h_n) Y_{2m}(u)},
	\label{beran}
\end{equation}
where $N_{2m}(u) = \delta_{2m} I(V_{2m} \leq u)$, $Y_{2m}(u) = I(R_m < u \leq V_{2m})$,
$\cK$ is the kernel function,
and $h_n$ is the bandwidth. Here, the estimator involves only the individuals $m$ for whom $T_{1m}$ was observed, but arguments given in the Supplementary Material
show that $\widehat{\Lambda}_{T_2|T_1}(t_2|t_1)$
still targets ${\Lambda}_{T_2|T_1}(t_2|t_1)$.
The corresponding conditional survival function estimate is 
$$
\widehat{S}_{T_2|T_1}(t_2|t_1) = \exp\left(-\widehat{\Lambda}_{T_2|T_1}(t_2|t_1)\right).
$$
Note that if $t_1 > R_U$, then automatically $S_{T_2|T_1}(r|t_1)=1$ for all
$r \in [R_L, R_U]$, so that $S_{T_2|T_1}(r|t_1)$ needs to be estimated only
for $t_1 \in [0,R_U]$.
Given the above, our proposed CIF estimator is then given by (\ref{est}) with $\widehat{Q}
= \widehat{B}^{-1}$, where $\widehat{B}$ is given by (\ref{qhat}). 

The asymptotic properties of the proposed estimator are described in the following theorem.

\begin{theorem}
Under the assumptions stated in the Supplementary Material, 
for given $t_1^*, \in [t_{1min},\tau]$, the estimation error 
 $\widehat{G}_1(t_1^*) - {G}_1(t_1^*)$ of the estimator $\widehat{G}_1(t_1^*)$
 can be represented in the form
\begin{equation}
\widehat{G}_1(t_1^*) - {G}_1(t_1^*) = \frac{1}{n} \sumi U_i(t_1^*) + o_{P}(n^{-1/2}),
\label{iidr}
\end{equation}
where $\{U_i(t_1^*)\}_{i=1}^n$ are i.i.d.\ mean-zero random variables
not depending on the bandwidth sequence $h_n$.
Consequently, we have
\begin{equation}
\sqrt{n} \, (\widehat{G}_1(t_1^*) - {G}_1(t_1^*))  \stackrel{d}{\to} N(0,\sigma^2(t_1^*))
\label{anrm}	
\end{equation} 
with $\sigma^2(t_1^*) = \Var(U_i(t_1^*))$.
\end{theorem}
Note that the theorem provides a pointwise result, not a stochastic process result.

In the Supplemental Material, we give an explicit formula for $U_i(t_1^*)$.
In principle, we could estimate $\sigma^2(t_1^*)$ by
$$
\widehat{\sigma}^2(t_1^*) = \frac{1}{n} \sumi \widehat{U}_i(t_1^*)^2, 
$$
where $\widehat{U}_i(t_1^*)$ is obtained by replacing unknown quantities
appearing in $U_i(t_1^*)$ by corresponding estimates. We could then
form a pointwise confidence interval for $G_1(t_1^*)$ as
$
\widehat{G}_1(t_1^*) \pm \zeta_\alpha \widehat{\sigma}^2(t_1^*),
$
with $\zeta_\alpha$ being the $(1-\alpha)$ quantile of the $N(0,1)$ distribution.
However, the expression for $U_i(t_1^*)$ is very complex, and the normal
approximation may not be satisfactory in finite samples.
We therefore propose instead a bootstrap-based procedure for pointwise
confidence intervals, which is presented in detail in Section~4.

In our numerical work, we use the triweight kernel
$$
\cK(\xi) = \frac{35}{32}(1-\xi^{2})^{3} I(|\xi|\leq 1).
$$
We implement a boundary correction at the left boundary using
the boundary kernel formula in Eqn.\ (3.4) of \citet{jones1993}.
For $\omega \in [0,1)$ define
$$
\mu_k(\omega) = \int_{-1}^\omega \xi^k \cK(\xi) d\xi
$$
Define $\omega(v_1,h) = (v_1-t_{1min})/h$.
We then replace $\cK(\xi)$ by 
$$
\cK_{bc}(v_1,\xi) = 
\begin{cases}
\scalebox{1.4}{$
\frac{(\mu_2(\omega(v_1,h))-\mu_1(\omega(v_1,h))\xi) \cK(\xi)}
{\mu_0(\omega(v_1,h))\mu_2(\omega(v_1,h)) - \mu_1(\omega(v_1,h))^2}$} 
& v_1 \in [t_{1min},t_{1min}+h) \\[2mm]
\cK(\xi) & v_1 \geq t_{1min} + h  
\end{cases}
$$

\section{Bandwidth Selection}

Our estimation procedure involves a bandwidth parameter $h_n$. The
the bandwidth enters only through $\widehat S_{T_2\mid T_1}$,
so we select $h_n$ by directly assessing predictive adequacy of the fitted 
conditional death distribution estimate.
We propose a $K$-fold cross-validation procedure based on martingale residuals.
The martingale residual \citep{ther1990} for observation $i$, adjusting for left truncation, is given by
$$
\widehat{\mathcal{M}}_i = \delta_{2i} - \widehat{\Pi}_i
$$
with
$$
\widehat{\Pi}_i = \widehat{\Lambda}_{T_2|T_1}(V_{2i}|V_{1i}) - \widehat{\Lambda}_{T_2|T_1}(R_i|V_{1i}).
$$
We divide the dataset into $K$ approximately equally-sized blocks $\cB_1, \ldots, \cB_K$. 
Consider a given bandwidth value $h$.
For each block $\mathcal{B}_k$, we compute 
the estimate
$$
\widehat{\Lambda}_{T_2|T_1}^\pss{-k} = \mbox{the estimate (\ref{beran}) based on the data excluding $\mathcal{B}_k$},
$$
and then for each $i \in \cB_k$, we compute 
$$
\widehat{\mathcal{M}}_i^\pss{-k} = \delta_{2i} - \widehat{\Pi}_i^\pss{-k}
$$
with 
$$
\widehat{\Pi}_i^\pss{-k} = \widehat{\Lambda}_{T_2|T_1}^\pss{-k}(V_{2i}|V_{1i}) 
- \widehat{\Lambda}_{T_2|T_1}^\pss{-k}(R_i|V_{1i}).
$$
Given the quantities defined above, we compute the goodness of fit statistic
$$
GOF_k(h) = \sum_{i \in \cB_k} I(\widehat{\Pi}_i^\pss{-k} > 0) \frac{(\widehat{\mathcal{M}}_i^\pss{-k})^2}{\widehat{\Pi}_i^\pss{-k}}.
$$
We then compute the overall goodness of fit statistic
$$
GOF(h) = \sum_{k=1}^K GOF_k(h).
$$
The final bandwidth is then given by the $h$ that minimizes $GOF(h)$.

\section{CIF Pointwise Confidence Intervals and Simultaneous Confidence Band}

We use a bootstrap scheme, with ordinary bootstrap sampling, to provide pointwise 
confidence intervals. We use a variation of what \citet[Section 2.4]{dh1997}
call the ``basic bootstrap confidence limits" on a transformed scale.
In preliminary simulations, we examined options for transformations, 
including no transformation, the arcsine transformation
and the complementary log-log transformation $g(u) = \log(-\log(1-u) + 10^{-8})$.
In the end, we settled on the the complementary log-log transformation,
which provided substantially better results than the other transformations.
For bootstrap sample $b$, $b = 1, \ldots, B$, let $\widehat{G}_{1b}(t_1^*)$ be the
estimate of $G_1(t_1^*)$. Define $\psi = g(G_1(t_1^*))$, 
$\widehat{\psi} = g(\widehat{G}_{1}(t_1^*))$,
$\widehat{\psi}_b = g(\widehat{G}_{1b}(t_1^*))$, and $\cT_b = \hat{\psi}_b - \hat{\psi}$.
The basic bootstrap confidence limits for $\psi$ at confidence level $1-\alpha$,
as presented by Davison and Hinkley, are given by
$\psi_L = \widehat{\psi} - q_2$ and $\psi_U = \widehat{\psi} - q_1$,	
where $q_1$ is the $\alpha/2$ quantile of the distribution of $\cT_b$ and
$q_2$ is the $(1-(\alpha/2))$ quantile.
Ordinarily, we would estimate $q_1$ and $q_2$ using empirical
quantiles of a very large bootstrap sample of $\cT_b$ values. 
In our setting, however, it is computationally impractical to generate a bootstrap 
sample with sufficiently many replications to provide a precise estimate of
the quantiles. We therefore apply a conservative adjustment to account
for the simulation error in estimating the quantiles.
Specifically, we replace $q_1$ by a lower $(1-\alpha)$ 
level confidence limit $\tilde{q}_1$ for $q_1$ and $q_2$ by an upper $(1-\alpha)$
confidence limit $\tilde{q}_2$ for $q_2$. Following the standard
nonparametric procedure for a confidence interval for a quantile,
we take $\tilde{q}_1$ to be the $k_1$-th
order statistic of the bootstrap sample $\cT_1, \ldots, \cT_B$ and $\tilde{q}_2$ 
to be the $k_2$-th order statistic, where
\begin{align*}
k_1 & = floor \left( \tfrac{1}{2}\alpha B - \zeta_\alpha 
\sqrt{B (\tfrac{1}{2}\alpha)\left(1-\tfrac{1}{2}\alpha\right)} \right) \\
k_2 & = ceiling \left( \left(1-\tfrac{1}{2}\alpha\right)B + \zeta_\alpha 
\sqrt{B(\tfrac{1}{2}\alpha)\left(1-\tfrac{1}{2}\alpha\right)} \right).
\end{align*}
Here, $floor(c)$ denotes the largest integer less than or equal to $c$
and $ceiling(c)$ denotes the smallest integer greater than or equal to $c$.
The confidence limits for $G_1(t_1^*)$ are then given by
$G_{1L}(t_1^*) = g^{-1} (\widehat{\psi} - \tilde{q}_2)$ and 
$G_{1U}(t_1^*) = g^{-1} (\widehat{\psi} - \tilde{q}_1)$,
where $g^{-1}$ is the inverse function of $g$, given by
$g^{-1}(a) = 1 - \exp(-(e^a-10^{-8}))$.

We also provide a simultaneous confidence band for $G_{1}(t_1^*)$. Define 
$\mathfrak{M}_b = \max_{t_1^*} |\widehat{G}_{1b}(t_1^*) - \widehat{G}_{1}(t_1^*)|$
and
$$
k^* = ceiling \left( (1-\alpha)B + \zeta \sqrt{B\alpha(1-\alpha)} \right)
$$
The confidence band is then given by $\widehat{G}_{1}(t_1^*) \pm \Delta$, where 
$\Delta$ is the $k^*$-th order statistic of the sample $\mathfrak{M}_1, \ldots,
\mathfrak{M}_B$.
We have not undertaken a formal theoretical analysis of this simultaneous confidence
band procedure, but our simulation results show good performance in many situations.

\omitt{
As part of the bootstrap procedure, we also provide a bias-corrected point
estimate, given by
$$
\widehat{G}_{1}^\pss{corr}(t_1^*)
= \widehat{G}_{1}(t_1^*)
- \left( \frac{1}{B} \sum_{b=1}^B \widehat{G}_{1b}(t_1^*) - \widehat{G}_{1}(t_1^*) \right)
=  2 \widehat{G}_{1}(t_1^*) - \frac{1}{B} \sum_{b=1}^B \widehat{G}_{1b}(t_1^*);
$$
see \citet[Section 10.6]{efron1993}.
}

\section{Simulation Study}
We conducted an extensive simulation study to examine the finite-sample properties of the
proposed estimator, $\widehat{G}_1$, in comparison with the Aalen-Johansen estimator,
$\widehat{G}_1^{AJ}$, with adjustment for left truncation, and the GZS
estimator $\widehat{G}_1^{GZS}$. 
Pointwise 95\% confidence intervals and were computed for all three estimators
using the procedure described in the preceding section.
In addition, 95\% simultaneous confidence bands were computed for the GZS and new estimators,
again using the procedure described in the preceding section.

We examined different combinations of the distribution 
of $T_1$, the distribution of $T_2$, and the distribution of $W$. The configurations are
labeled with a three-digit identifier, with the three digits representing, respectively, 
the distribution of $T_1$, $T_2$, and $W$. The choices for $T_1$ were as follows:

1: truncated Weibull with shape parameter 4, scale parameter 116, and lower truncation point 40 

2: 20 + Weibull with shape parameter 0.85 and scale parameter 170 

3: Weibull with shape parameter 3.5 and scale parameter 100

The sampling scheme was as follows: Initially, we generated a disease time $T_1^{(\text{ini})}$ according the
selected distribution (from the three options indicated above) and a death time $T_2^{(\text{ini})}$ based on mortality data 
from the UK Office of National Statistics. If $T_2^{(\text{ini})} < T_1^{(\text{ini})}$, we set $T_2 = T_2^{(\text{ini})}$ and 
$T_{1i}=\infty$. If $T_2^{(\text{ini})} > T_1^{(\text{ini})}$, we set $T_1 = T_1^{(\text{ini})}$ and generated $T_2$ according 
to the following options:

1: $T_2 = T_1 + 0.8(T_2^{(\text{ini})} - T_1)$ 

2: $T_2 = T_1 + U$, where $U$ follows a Weibull distribution with a shape parameter of 4 and a scale parameter chosen to achieve an expected survival time after diagnosis of 5 years 

3: $T_2 = T_1 + U$, where $U$ follows a Weibull distribution with a shape parameter of 4 and a scale parameter chosen to achieve an expected survival time after diagnosis of 15 years

The options for $W$ were as follows:

1: uniform over [11,15] (similar to UKB data) 

2: uniform over [11,25] 

The recruitment time $R$ was generated according to the distribution of recruitment times in the UKB data.
The minimum recruitment age $R_L$ was 40 and the maximum recruitment age $R_U$ was 69.
The first option for the distribution of $T_1$ corresponds to scenarios where no disease events
occurred before $R_L$, while the other two options correspond to scenarios
where the onset of disease could occur before $R_L$.

The simulations were carried out using parallel computation with 32 parallel processors.
Each configuration was studied with 512 simulation replications and 500 bootstrap samples, with
sample sizes of 2,500, 5,000 and 7,500. Table 1 summarizes the censoring rates, the death
rates before and after disease diagnosis, and the percentage of cases that were prevalent
cases. 

Figures 1-3  summarize the results for the nine scenarios with $W$ following the $U(11,15)$ distribution
and Figures 4-6 the results for the nine scenarios with $W$ following the $U(11,25)$ distribution.
The plots show the results of the Aalen-Johansen estimator, the GZS estimator, and the new proposed estimator, 
with $n=5,000$ observations. The age range in the plots is 40 to 80 in the scenarios with no disease events
before $R_L$ and 30 to 80 in the remaining scenarios.
The graphs in the figures present the mean, empirical standard deviation, and pointwise CI coverage rate 
of the estimators over the 512 simulation replications, as a function of $t$.
Section S2 of the Supplementary Material includes results for $n=2,500$ and $n=7,500$.

Regarding the scenarios with shorter follow-up (xx1), in Scenarios 111, 121, 211, 221,
311, and 321, the GZS estimator exhibited serious bias in the mid-to-upper end of the 
age range. The new estimate
was on target in all scenarios except for 221 and 321, in which
the new estimate was slightly biased.
The bias with the new estimator in Scenarios 221 and 321 decreased with increasing
sample size.
The pointwise confidence interval coverage with the new estimator was
satisfactory in most scenarios.
The simultaneous confidence band coverage with the new estimator was 
satisfactory except in Scenarios 221 and 321.
In the scenarios with longer follow-up (xx2), the GZS estimator was no
longer biased, but the new estimator had lower standard deviation.
The pointwise confidence interval coverage and simultaneous confidence band coverage
with both the GZS estimator and the new estimator was satisfactory in most scenarios.

\section{UK Biobank Data Analysis}
The UK Biobank (UKB) provides extensive data on approximately 500,000 individuals in the UK, who were aged 40–69 at enrollment between 2006 and 2010 and are followed prospectively until death, dropout, or the end of the study. In the present analysis, we consider 10 phenotypes based on approximately 250,000 female participants. Figure \ref{fig:UKB-interesting} presents CIF estimates for six cancers—breast, cervical, colorectal, ovarian, melanoma, and kidney—obtained using three methods: Aalen–Johansen (AJ), GZS, and the proposed method (NEW). Table 3 reports, for each disease, the numbers of prevalent and incident cases, together with the widths of the confidence bands and the mean widths of the pointwise confidence intervals under each method and based on 500 bootstrap samples.

Evidently, for all six cancers, the CIF estimates produced by the left-truncation adjusted AJ estimator and the proposed estimator are very similar, whereas the GZS estimator systematically underestimates the CIF. This downward bias is expected in settings with relatively long survival after diagnosis, because the GZS estimator assigns positive mass only to disease cases whose subsequent death is observed during follow-up. When follow-up is not sufficiently long to observe enough post-diagnosis deaths, the resulting estimator fails to fully recover the incidence curve. Thus, in these UK Biobank applications, AJ and the proposed estimator appear to target the same incidence curve, while GZS is not adequate for diseases with long post-diagnosis survival. A key advantage of the proposed estimator over AJ is efficiency: by incorporating all observed disease cases, including prevalent cases and cases not followed until death, the proposed method yields uniformly narrower pointwise confidence intervals and simultaneous confidence bands across all phenotypes. This efficiency gain is substantial and represents the central practical benefit of the proposed approach in settings where AJ is approximately unbiased but discards a large fraction of informative disease cases. For breast cancer, for example, the average width of the 95\% simultaneous confidence band is reduced from 0.0068 under AJ to 0.0035 under the proposed estimator, corresponding to a nearly two-fold reduction.

In the Supplementary Material, Figure S5 and Table S1 present analogous results for four additional cancers: brain, esophageal, acute myeloid leukemia, and lung. For these phenotypes, the CIF estimators from all three methods are highly similar, likely because the interval between diagnosis and death is shorter, so deaths after diagnosis are observed more completely.

\section{Concluding Remarks}

In this paper, we present a new CIF estimator that incorporates information from all disease cases, including both incident cases and prevalent cases,
irrespective of whether or not the diseased individual died during the course of the study. This goes beyond the estimator of \cite{aalen1978empirical} (AJ),
which incorporates only incident cases, and therefore cannot estimate the CIF for ages prior to the minimum age of recruitment to the study. It also goes
beyond the estimator of \cite{gzs2025} (GZS), which also incorporates prevalent cases in which the diseased individual died during the study, but not cases 
in which the disease individual survived to the end of follow-up. For diseases, such as breast cancer, which can occur early in life and are associated 
with long survival after disease onset, the GZS estimator is highly biased if the length of follow-up is short to moderate, whereas the new estimator
is consistent. Even when follow-up is long, so that the GZS estimator is on target, the new method has the advantage of a lower standard deviation.
The new estimator occasionally exhibits finite-sample bias, but the bias is very slight.

\section*{Acknowledgments}
This research has been conducted using the UK Biobank Resource, project 56885.
We thank Ingrid van Keilegom for directing us to \cite{vankeil1999}.
M.G. was supported by ISF grant 767/21 and a competitive
grant in data science (DS) from the Israel Council of Higher Education (Malag).

\section*{Supplementary Material}

Supplementary material is available at \textit{Biostatistics} online.

\section*{Conflicts of Interest}
No conflicts of interest are declared.

\section*{Software Availability}

The R code used for carrying out the simulations reported in this paper is available on
Github at the following address: \\
\url{https://github.com/david-zucker/illness-death-1} \\
An R package \texttt{illdthCIF1} implementing the proposed method is available on Github
at the following address: \\
\url{https://github.com/david-zucker/illdthCIF1}
\bibliography{research-program}

\begin{figure}[H]
	\centering
	\includegraphics[width=1\textwidth, height=0.85\textheight]{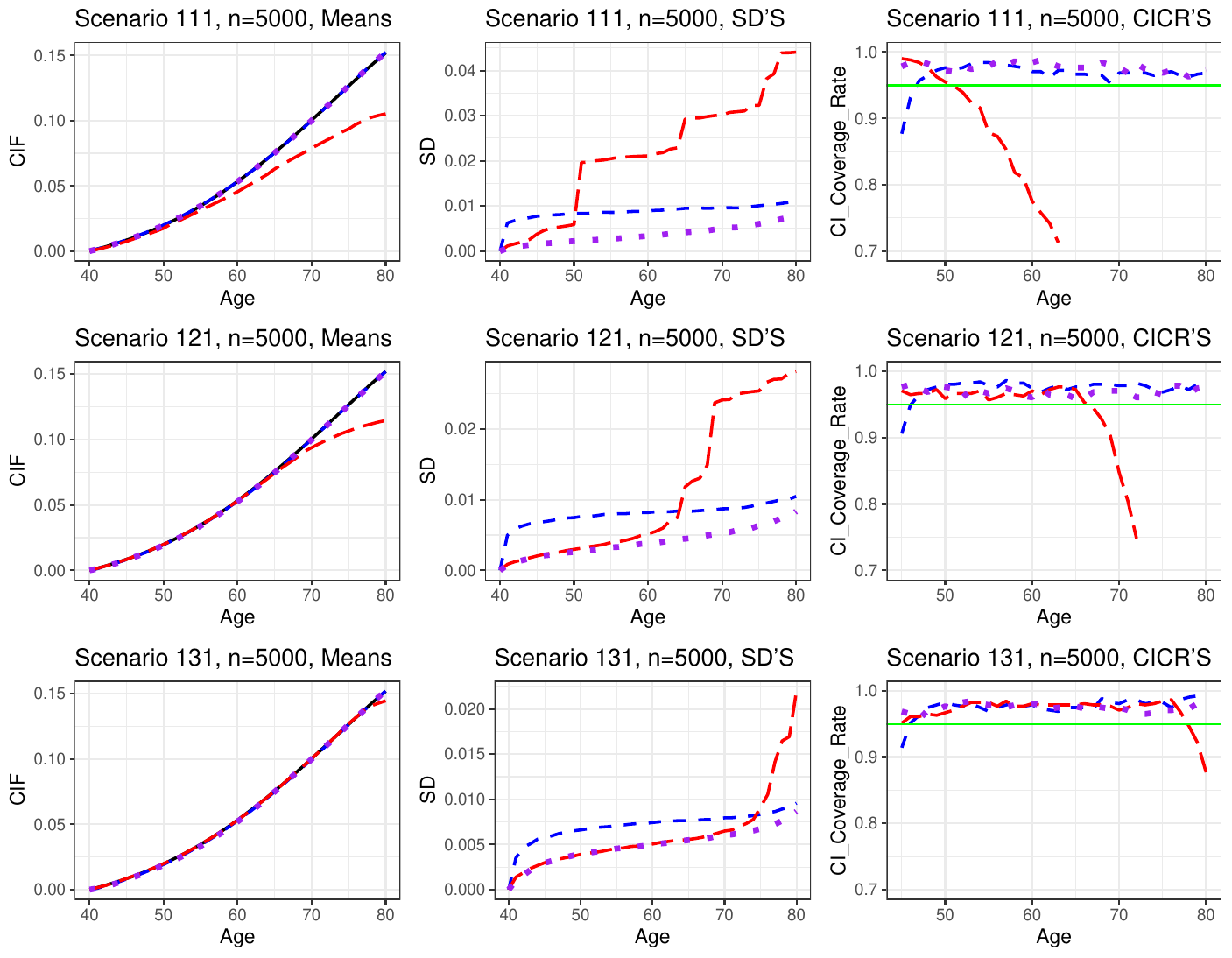}
	\caption{Simulation results for configurations 111, 121, and 131: Mean over estimates, standard deviations (SD), and 
	   empirical coverage rates of 95\% pointwise confidence intervals, for the Aalen-Johansen (AJ), Gorfine-Zucker-Shoham (GZS),
	   and new estimators. Solid black line – true curve, blue dashed line – AJ estimator, red long-dashed line - GZS estimator, 
	   purple dotted line – new estimator. }\label{fig:f1} 
\end{figure}

\begin{figure}[H]
	\centering
	\includegraphics[width=1\textwidth, height=0.85\textheight]{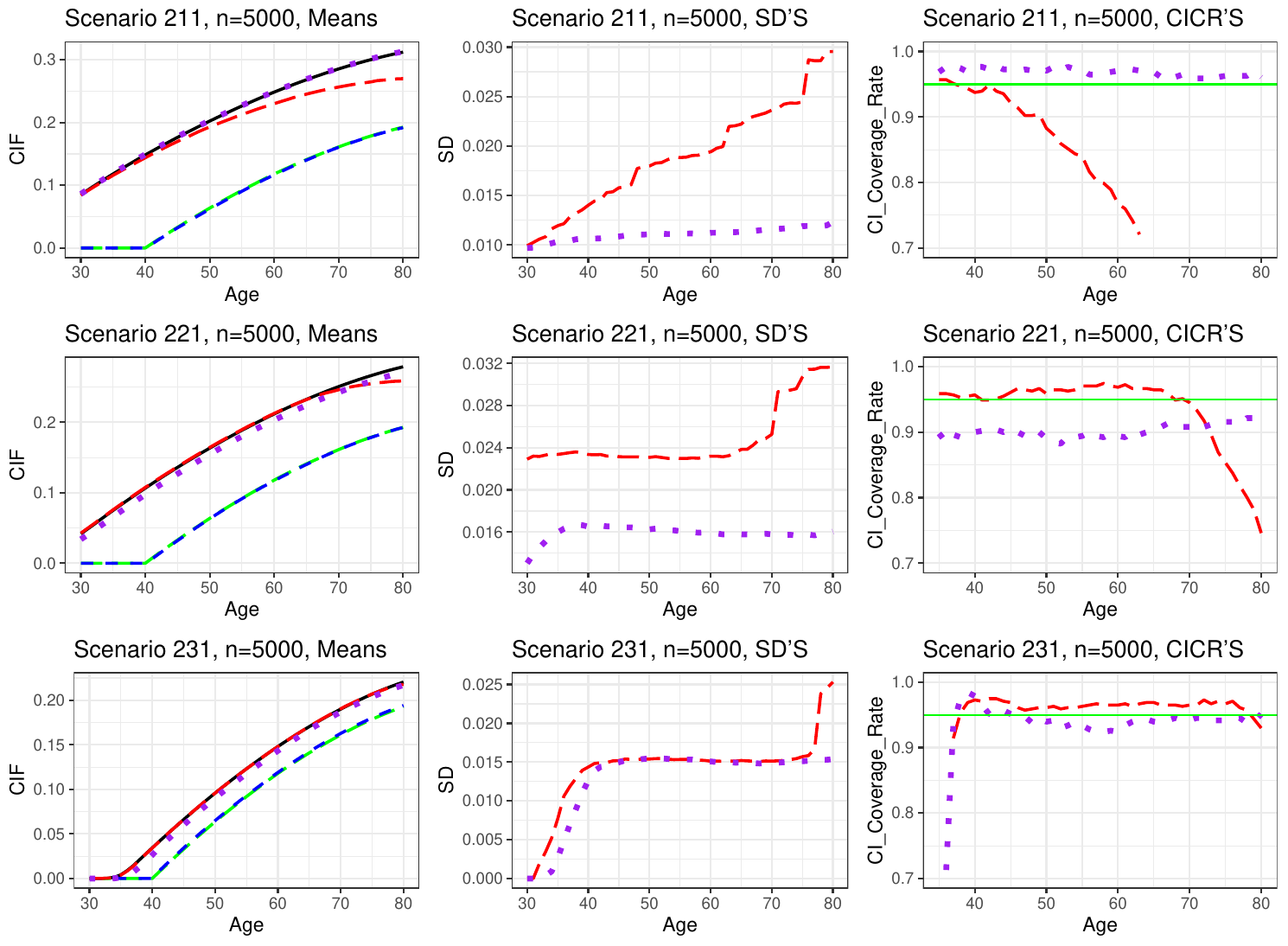}
	\caption{Simulation results for configurations 211, 221, and 231: Mean over estimates, standard deviations (SD), and 
		empirical coverage rates of 95\% pointwise confidence intervals for the Gorfine-Zucker-Shoham (GZS) and 
		and new estimators. Mean over estimates and target CIF curve shown also for the Aalen-Johansen (AJ) estimator.
		 Solid black line – true curve, green long-dashed line - target of AJ estimator,
		blue dashed line – AJ estimator, red long-dashed line - GZS estimator, 
		purple dotted line – new estimator. }\label{fig:f2} 
\end{figure}

\begin{figure}[H]
	\centering
	\includegraphics[width=1\textwidth, height=0.85\textheight]{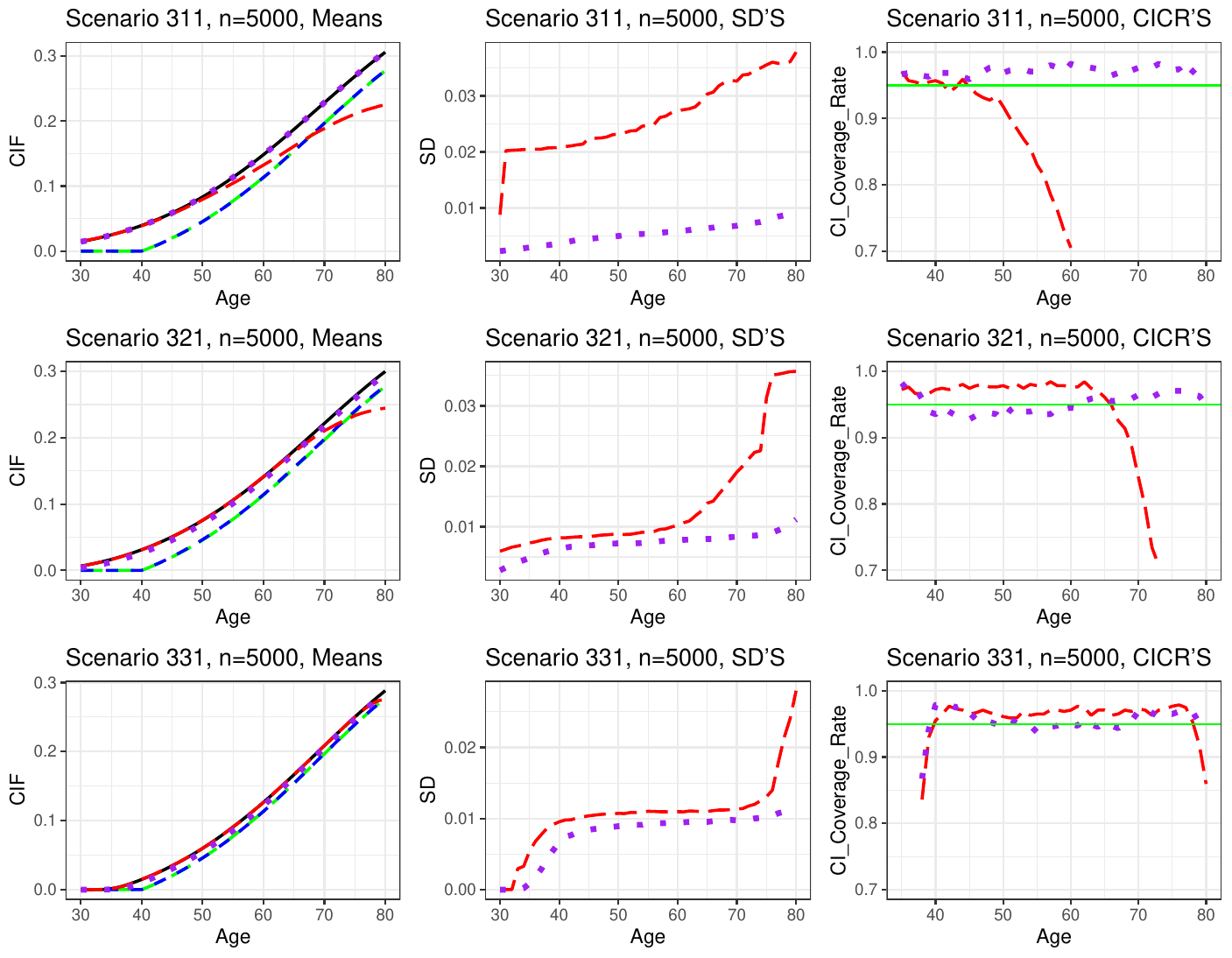}
	\caption{Simulation results for configurations 211, 221, and 231: Mean over estimates, standard deviations (SD), and 
		empirical coverage rates of 95\% pointwise confidence intervals for the Gorfine-Zucker-Shoham (GZS) and 
		and new estimators. Mean over estimates and target CIF curve shown also for the Aalen-Johansen (AJ) estimator.
		Solid black line – true curve, green long-dashed line - target of AJ estimator,
		blue dashed line – AJ estimator, red long-dashed line - GZS estimator,  
		purple dotted line – new estimator. }\label{fig:f3} 
\end{figure}

\begin{figure}[H]
	\centering
	\includegraphics[width=1\textwidth, height=0.85\textheight]{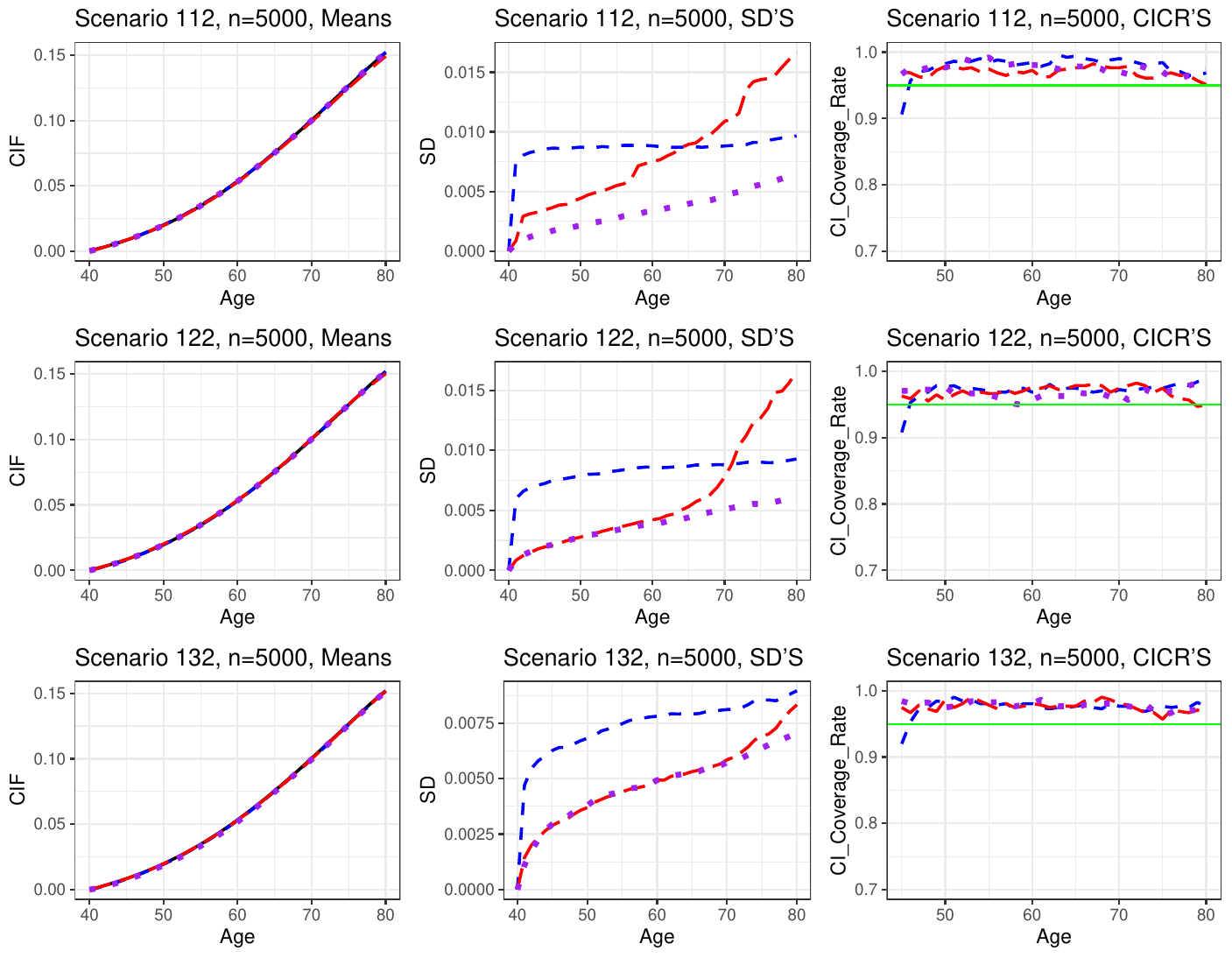}
	\caption{Simulation results for configurations 112, 122, and 132: Mean over estimates, standard deviations (SD), and 
		empirical coverage rates of 95\% pointwise confidence intervals, for the Aalen-Johansen (AJ), Gorfine-Zucker-Shoham (GZS),
		and new estimators. Solid black line – true curve, blue dashed line – AJ estimator, red long-dashed line - GZS estimator, 
		purple dotted line – new estimator. }\label{fig:f4} 
\end{figure}

\begin{figure}[H]
	\centering
	\includegraphics[width=1\textwidth, height=0.85\textheight]{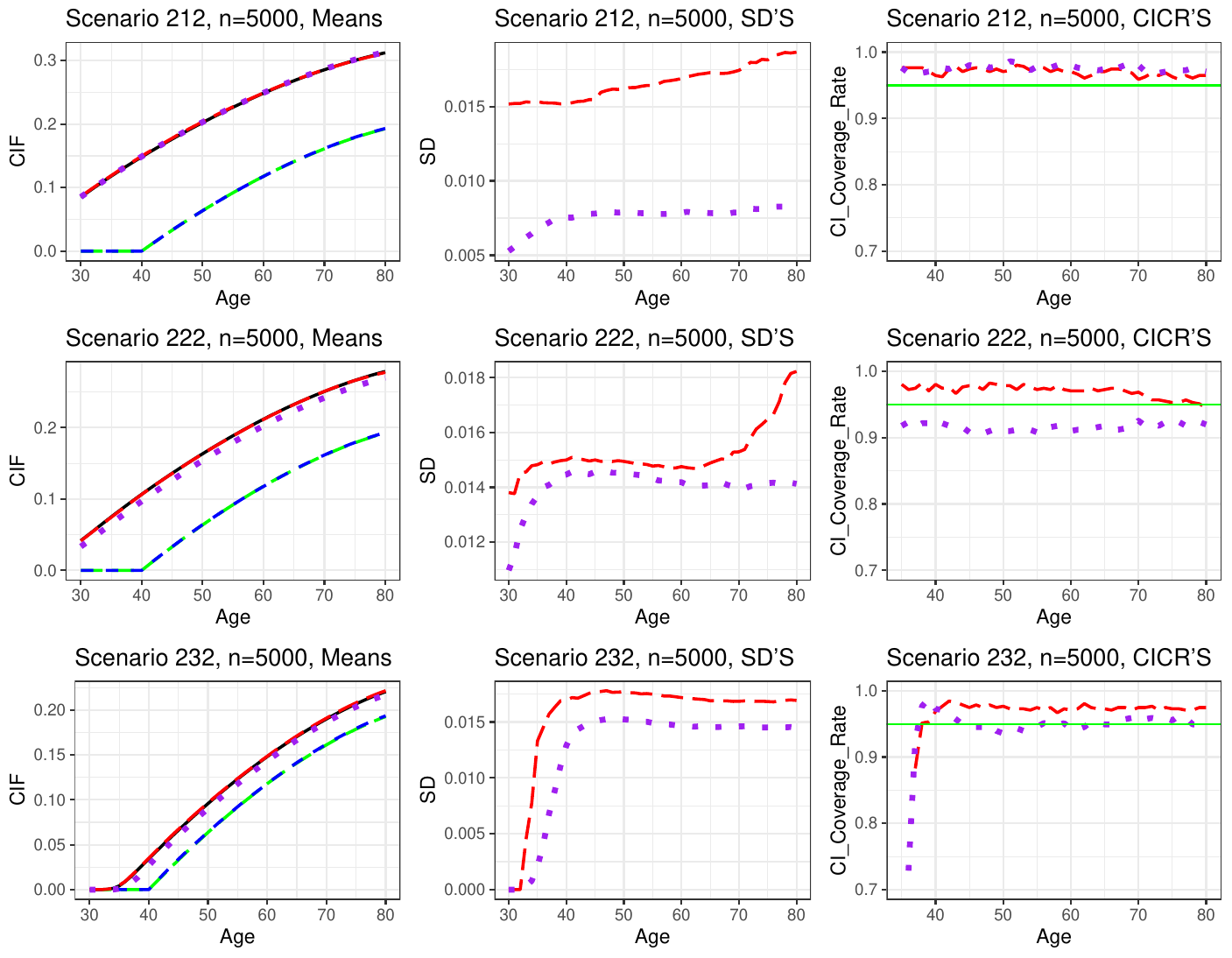}
	\caption{Simulation results for configurations 212, 222, and 232: Mean over estimates, standard deviations (SD), and 
		empirical coverage rates of 95\% pointwise confidence intervals for the Gorfine-Zucker-Shoham (GZS) and 
		and new estimators. Mean over estimates and target CIF curve shown also for the Aalen-Johansen (AJ) estimator.
		Solid black line – true curve, green long-dashed line - target of AJ estimator,
		blue dashed line – AJ estimator, red long-dashed line - GZS estimator, 
		purple dotted line – new estimator. }\label{fig:f5} 
\end{figure}

\begin{figure}[H]
	\centering
	\includegraphics[width=1\textwidth, height=0.85\textheight]{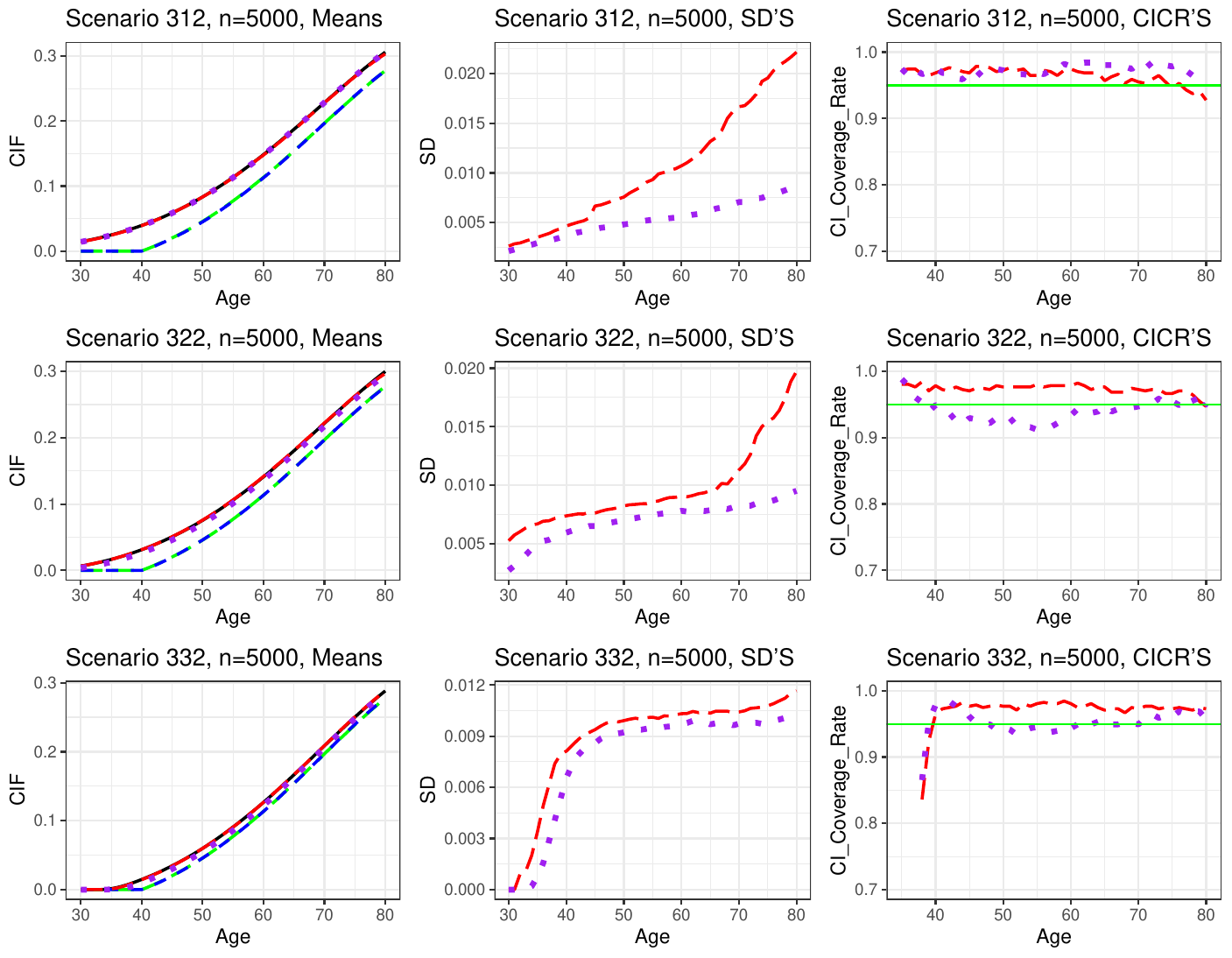}
	\caption{Simulation results for configurations 212, 222, and 232: Mean over estimates, standard deviations (SD), and 
		empirical coverage rates of 95\% pointwise confidence intervals for the Gorfine-Zucker-Shoham (GZS) and 
		and new estimators. Mean over estimates and target CIF curve shown also for the Aalen-Johansen (AJ) estimator.
		Solid black line – true curve, green long-dashed line - target of AJ estimator,
		blue dashed line – AJ estimator, red long-dashed line - GZS estimator, 
		purple dotted line – new estimator. }\label{fig:f5} 
\end{figure}

\begin{figure}[H]
	\centering
	\includegraphics[width=1\textwidth, height=0.65\textheight]{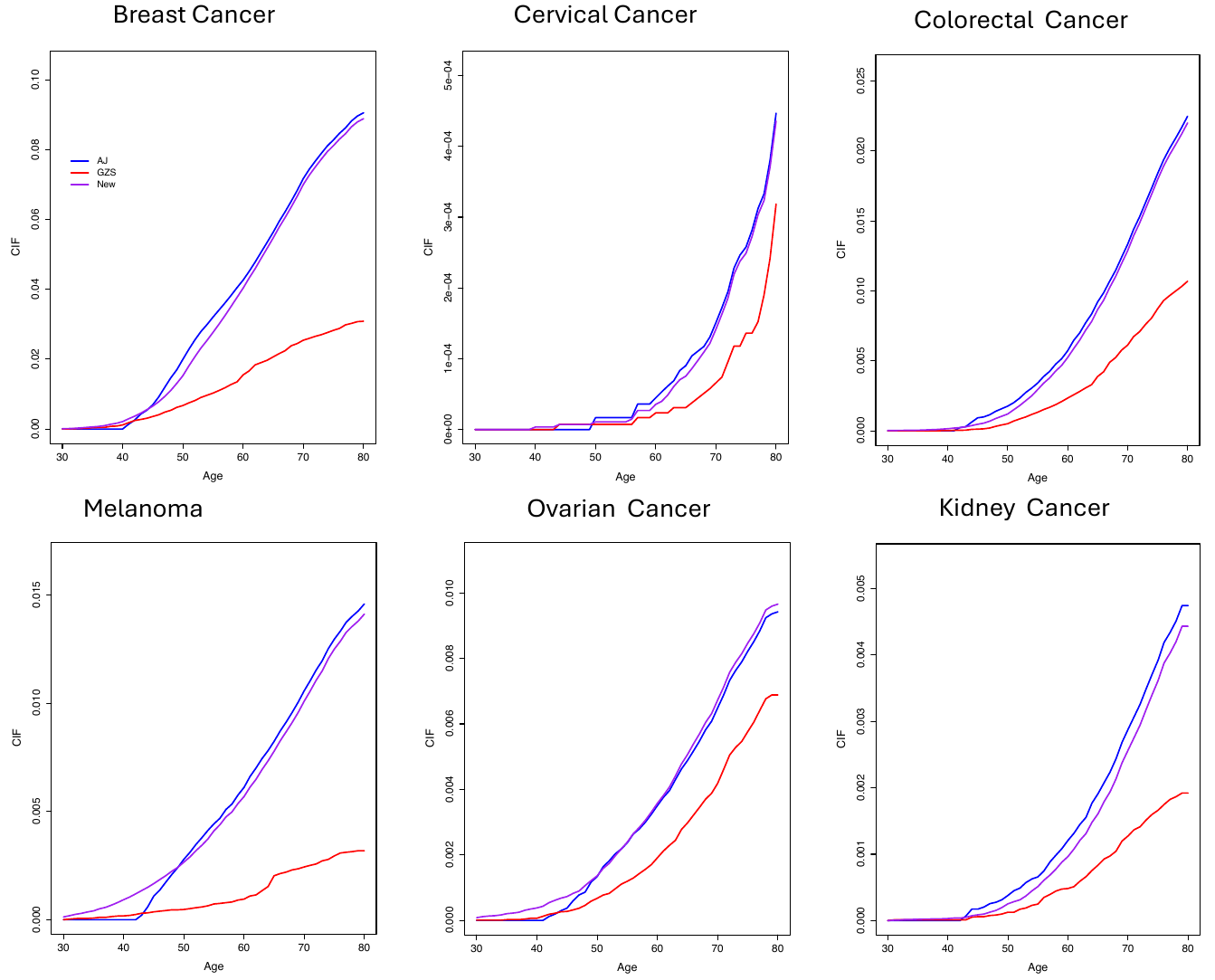}
\caption{Estimated CIFs for six phenotypes in female participants from the UK Biobank, obtained using the Aalen--Johansen estimator (AJ, blue), the estimator of \cite{gzs2025} (GZS, red), and the proposed estimator (New, purple).}\label{fig:UKB-interesting}
\end{figure}

\begin{table}[H]
	\centering
	\caption{Simulation Settings - Summary Statistics}\label{tbl:Simulation-info}
	\begin{tabular}{|r|rrrr|r|}
		\hline
		& & Died Without & Alive With & Died With & Prevalent \\
		Scenario & Censored & Disease & Disease & Disease & Fraction \\
		\hline
		111 & 73.5 & 16.4 &  7.5 &  2.7 & 41.3 \\
		112 & 60.4 & 26.9 &  7.3 &  5.4 & 33.1 \\
		121 & 74.0 & 16.5 &  5.7 &  3.9 & 37.0 \\ 
		122 & 60.8 & 27.1 &  5.6 &  6.5 & 29.4 \\ 
		131 & 75.4 & 16.8 &  2.4 &  5.4 & 21.6 \\
		132 & 62.0 & 27.6 &  2.3 &  8.1 & 16.2 \\
		211 & 59.4 & 13.2 & 17.9 &  9.5 & 79.6 \\
		212 & 49.2 & 21.8 & 14.0 & 15.0 & 75.3 \\
		221 & 68.7 & 15.3 &  6.8 &  9.2 & 59.6 \\
		222 & 56.8 & 25.3 &  5.6 & 12.2 & 53.5 \\
		231 & 73.5 & 16.3 &  2.3 &  7.8 & 31.4 \\
		232 & 60.8 & 27.1 &  1.9 & 10.2 & 26.4 \\
		311 & 63.2 & 14.1 & 16.1 &  6.6 & 54.2 \\
		312 & 50.8 & 22.6 & 14.5 & 12.1 & 46.2 \\
		321 & 66.1 & 14.7 & 10.5 &  8.6 & 43.5 \\
		322 & 53.2 & 23.6 &  9.8 & 13.4 & 35.9 \\
		331 & 69.5 & 15.5 &  4.3 & 10.8 & 23.8 \\
		332 & 55.8 & 24.9 &  3.9 & 15.4 & 18.5 \\
		\hline
	\end{tabular}  
\end{table}

\begin{table}[H]
	\centering
	\caption{Simultaneous Confidence Band Coverage Rates for the GZS and New Estimators}\label{tbl:cbands}
	\begin{tabular}{|r|rr|rr|rr|}
		\hline
        & \multicolumn{2}{c|}{$n=2,500$}  
		& \multicolumn{2}{c|}{$n=5,000$} 
		& \multicolumn{2}{c|}{$n=7,500$} \\
        Scenario & GZS & New & GZS & New & GZS & New \\
		\hline
		111 & 28.9 & 95.9 & 24.2 & 96.5 & 20.5 & 96.1 \\
		112 & 93.6 & 95.9 & 93.2 & 94.7 & 91.4 & 97.3 \\
		121 & 38.1 & 96.5 & 25.2 & 95.7 & 21.3 & 95.9 \\ 
		122 & 92.4 & 94.7 & 91.8 & 96.5 & 89.6 & 97.3 \\
		131 & 85.4 & 96.9 & 83.0 & 96.7 & 82.6 & 95.3 \\
		132 & 97.5 & 96.7 & 95.5 & 96.5 & 96.9 & 96.7 \\
		211 & 53.9 & 97.9 & 39.3 & 96.5 & 35.0 & 96.7 \\
		212 & 95.9 & 97.7 & 95.1 & 97.3 & 93.9 & 97.7 \\
		221 & 78.3 & 78.9 & 72.9 & 82.8 & 64.5 & 82.0 \\
		222 & 93.9 & 78.3 & 94.1 & 83.0 & 93.4 & 81.6 \\
		231 & 92.4 & 86.7 & 91.6 & 88.1 & 91.4 & 89.1 \\
		232 & 92.8 & 85.5 & 94.1 & 88.9 & 95.5 & 87.5 \\
		311 & 30.3 & 96.5 & 18.2 & 97.9 & 13.5 & 96.7 \\
		312 & 92.4 & 96.9 & 92.2 & 97.1 & 89.8 & 97.1 \\
		321 & 34.6 & 93.4 & 26.6 & 93.4 & 18.2 & 93.0 \\
		322 & 93.2 & 93.0 & 94.3 & 91.0 & 92.0 & 93.4 \\
		331 & 86.3 & 93.7 & 82.8 & 93.6 & 82.2 & 95.7 \\
		332 & 97.1 & 90.8 & 96.7 & 92.2 & 96.7 & 94.3 \\
		\hline
	\end{tabular}  
\end{table}

\begin{table}[H]
	\centering
	\caption{Results of 95\% confidence band width (Band) and mean 95\% pointwise confidence interval width (Pointwise) for six phenotypes in female participants from the UK Biobank, obtained using the Aalen--Johansen estimator (AJ), the estimator of \cite{gzs2025} (GZS), and the proposed estimator (New)}\label{tbl:UKB-interesting}
	\begin{tabular}{|rcc|cc|cc|cc|}
	\hline
	& & & \multicolumn{2}{|c|}{AJ} & \multicolumn{2}{|c|}{GZS} & \multicolumn{2}{|c|}{New}  \\
    Cancer type & prevalent & incident & Band & Pointwise & Band & Pointwise & Band & Pointwise \\
    \hline
    Breast  & 8729 & 8256 & 0.0068 & 0.0243 & 0.0067 & 0.0083 & {\bf 0.0035} & {\bf 0.0017}  \\
    Cervical  & 8  &  41  & 0.0004 & 0.0146 &  0.0004 & 0.0116 &  {\bf 0.0003} &  {\bf 0.0025} \\
    Colorectal & 1009 & 966 & 0.0026 & 0.0405 & 0.0021 & 0.0078 & {\bf 0.0020} & {\bf  0.0007} \\
    Ovarian & 664 & 966 &  0.0016 & 0.0347 &  0.0016 &  0.0069 &  {\bf 0.0012} &  {\bf 0.0006} \\
    Melanoma & 1272 & 1427 & 0.0022 & 0.0207 & 0.0023 &  0.0071 & {\bf 0.0014} & {\bf 0.0006}  \\
    Kidney  &184 & 510 & 0.0011 & 0.0646 & 0.0008 &   0.0288 & {\bf 0.0008}&  {\bf 0.0017}  \\
    \hline
	\end{tabular}  
\end{table}

\newpage 

\renewcommand\thesection{S\arabic{section}}
\renewcommand{\thetable}{S\arabic{table}}
\renewcommand{\thefigure}{S\arabic{figure}}
\renewcommand{\theequation}{S.\arabic{equation}}

\setcounter{section}{0}
\setcounter{table}{0}
\setcounter{figure}{0}
\setcounter{equation}{0}

\LARGE

\begin{center}
	
\vspace*{10mm}	

{\bf Supplementary material for \\ ``Efficient Cumulative Incidence Estimation in Biobank Studies Using All Prevalent and Incident Events"}

\vspace*{10mm}	

David M. Zucker\\ 
Department of Statistics and Data Science\\
The Hebrew University of Jerusalem, Israel\\ 		

\vspace*{10mm}	

Malka Gorfine \\
Department of Statistics and Operations Research\\ Tel Aviv University, Israel	

\end{center}

\normalsize

\newpage

\section{Proof of Theorem 1}

\noindent
\underline{Preliminaries}

Recall the definition of $\cKbc(v_1,\xi)$:
$$
\cKbc(v_1,\xi) = 
\begin{cases}
	\scalebox{1.4}{$
		\frac{(\mu_2(\omega(v_1,h))-\mu_1(\omega(v_1,h)\xi) \cK(\xi)}
		{\mu_0(\omega(v_1,h))\mu_2(\omega(v_1,h)) - \mu_1(\omega(v_1,h))^2}$} 
	& v_1 \in [t_{1min},t_{1min}+h) \\[2mm]
	\cK(\xi) & v_1 \geq t_{1min} + h  
\end{cases}
$$
with
$$
\mu_k(\omega) = \int_{-1}^\omega \xi^k \cK(\xi) d\xi
$$
and $\omega(v_1,h) = (v_1-t_{1min})/h$.

Recall also the following definition:
$$
B(v_1) = \int_{R_L}^{R_U} f_{R|R \leq T_2}(r) \cA(r,v_1) \sttw(r|v_1) \, dr, \\
$$
Let us introduce the following additional definitions:
\begin{align*}
	& \cA(r,v_1) = \sttwo(r)^{-1} S_W((v_1-r)_+), \\
	& A_k(v_1,u,h) = \frac{1}{nh} \summ Y_{2m}(u) Z_k\left(v_1,\frac{V_{1m}-v_1}{h}\right), \quad Z_k(v_1,q) = q^k \cKbc(v_1,q), \\
	& \cI =  [t_{1min}+h,R_U], \quad \cU = [t_{1min},t_{1min} + h), \\
	& \eta_k(v_1,h) = \int_{-v_1/h}^{(\tau-v_1)/h} Z_k(v_1,\rho) d\rho, \\
	& \varphi(u,v_1)= f_{V_{1m}}(v_1) E[Y_{2m}(u)|V_{1m}=v_1], \\
	& a_k(v_1,u,h) = \eta_k(v_1,h) \varphi(u,v_1)
\end{align*}
Note that
\begin{itemize}
	\item $\eta_0(v_1,h)=1$ for all $v_1$ and $h$, 
	\item $\eta_1(v_1,h)=a_1(v_1,h)=0$ for all $v_1$ and $h$,
	\item $\eta_k(v_1,h) = \mu_k(1)$ for $v_1 > t_{1min} + h$. 
\end{itemize}
Also define $\tau^* = \tau - R_L$.

\newpage

\underline{Assumptions}

The proof will be carried out under the following assumptions.

A1. The kernel $\cK$ is symmetric, equal to zero outside of $[-1,1]$, and equal
to a polynomial inside $[-1,1]$. In addition, $\cK$ is twice differentiable with
bounded derivatives over the entire real line, including the points $-1$ and 1.

A2. $h_n = \alpha_n n^{-\nu}$ with $\nu \in (\tfrac{1}{4},\tfrac{1}{3})$ and $\alpha_n \to \alpha > 0$. 

A3. The random vector $(T_1, T_2)$ has has a bounded joint density function on $[0,\tau]^2$,
the random variable $R$ has a bounded density on $[R_L,R_U]$, and the random variable
$W$ has a bounded density on $[0,\tau^*]$.
In addition, the density of $T_1$ is differentiable and bounded below by a positive 
number $f_{T_1}^*$ on the interval $[t_{1min},R_U+\varepsilon]$ for some $\varepsilon>0$.

A4. $B(v_1)$ is bounded below by a positive number $B^*$.

A5. The function $\lambda_{T_2\mid T_1}(t_2\mid t_1)$
is uniformly bounded over
$t_1\in[t_{1min},R_U+\varepsilon]$ and $t_2\in[0,R_U]$.
In addition, 
$\lambda_{T_2\mid T_1}(t_2\mid t_1)$ is twice continuously differentiable in $t_1$
on $t_1\in[t_{1min},R_U+\varepsilon]$ for all in $t_2\in[0,R_U]$, with derivatives that are uniformly bounded over $t_1$ and $t_2$.

A6. $\sttwo(\tau^-) > 0$ and $S_W((\tau^{*-}) > 0$.

A7. $\inf_{v_1 \in [0,R_U+\varepsilon]} S_{T_2|T_1}(R_U|v_1) \geq s_{min} > 0$.


Note that Assumptions A3 and A6 imply that the mapping $r\mapsto \sttwo(r)^{-1}$ is Lipschitz
and that the mappings $r\mapsto S_W((v_1-r)_+)$ and $r \mapsto \sttw(t|v_1)$ are uniformly Lipschitz over $v_1$. 
They also imply that $\inf_{t_1\in[t_{1\min},U_R],\,u\in[t_1,\tau]} a_0(t_1,u) > 0$.

\underline{Supporting Propositions}

We will need the following supporting propositions.

\begin{proposition} 
	(Lemma 8.1 of \cite{ejs2019}, slightly modified).
	We have
	\begin{align*}
		E[A_0(v_1,u,h)] & = \varphi(u,v_1) + O(h^2),	\\
		E[A_1(v_1,u,h)] & = O(h),	
	\end{align*}
	for all $v_1, u$, and $h$.
	
	For $k \geq 2$ even we have
	$$
	E[A_k(v_1,u,h)] = \left\{
	\begin{array}{ll}
		a_k(v_1,u,h) + O(h^2) & v_1 \in \cI \\
		a_k(v_1,u,h) + O(h) & v_1 \in \cU
	\end{array} \right.
	$$
	\em{and for $k \geq 3$ odd we have}\rm
	$$
	E[A_k(v_1,u,h)] = \left\{
	\begin{array}{ll}
		O(h) & v_1 \in \cI \\
		a_k(v_1,u,h) + O(h) & v_1 \in \cU
	\end{array} \right.
	$$
	\em{In addition, for any $k$,}\rm
	$$
	\sup_{u,v_1} |A_k(v_1,u,h) - E[A_k(v_1,u,h)]| = \Op\left((nh)^{-1/2} (\log n)^{1/2}\right)
	$$
	\em{In general, we have}\rm
	\begin{align*}
		\sup_{u,v_1} |A_k(v_1,u,h) - a_k(v_1,u,h)| \leq
		&\sup_{u,v_1} |A_k(v_1,u,h) - E[A_k(v_1,u,h)]| \\ 
		& \hspace*{24pt} + \sup_{u,v_1} |E[A_k(v_1,u,h)] - a_k(v_1,u,h)|
	\end{align*}
	\em{When $k=0$, the first term dominates, so that}
	$$
	\sup_{u,v_1} |A_k(v_1,u,h) - a_k(v_1,u,h)| \leq  \Op\left((nh)^{-1/2} (\log n)^{1/2}\right).
	$$
	\em{When $k$ is an even number greater than 0, the first term dominates for $v_1 \in I$ while the second term dominates for $v_1 \in \cU$, so that we obtain}\rm
	$$
	\sup_{u,v_1} |A_k(v_1,u,h) - a_k(v_1,u,h)|
	= \left\{
	\begin{array}{ll}
		\Op(n^{-(1-\nu)/2} \, (\log n)^{1/2}) & v_1 \in \cI \\
		\Op(n^{-\nu}) & v_1 \in \cU
	\end{array}
	\right.
	$$
	\rm{When $k$ is odd, we get}\rm
	$$
	\begin{array}{ll}
		\sup_{u,v_1} |A_k(v_1,u,h)| = \Op(h) & v_1 \in \cI \\
		\sup_{u,v_1} |A_k(v_1,u,h) - a_k(v_1,u,h)| =  \Op(h) & v_1 \in \cU
	\end{array}
	$$
\end{proposition} 

\vspace*{4mm}

\begin{proposition}
	(Example 19.7 of \cite{vdv1998}).
	Let $X$ be a random variable on a probability space $(\Omega, \cF, P)$.	
	Let $\mathscr{F} = \{f_\theta, \, \theta \in \Theta\}$ be a class of measurable functions from $\real$ to $\real$
	indexed by a bounded set $\Theta \in \real^d$. Suppose there exists a measurable function $m$ from $\real$ to 
	$[0,\infty)$ such that
	$$
	|f_{\theta_1}(x) -  f_{\theta_2}(x)|
	\leq m(x) \| \theta_1 - \theta_2 \|
	$$
	and $E[|X|^\ell] < \infty$ for some $\ell \geq 1$. Then $\mathscr{F}$ is $P$-Donsker.
\end{proposition}

\vspace*{4mm}

\begin{proposition} Let $\cD$ denote the finite rectangle $[a_1,b_1] \times [a_2,b_2]$ 
	and let $g(x,y)$ be a function from $\cD$ to $\real$ which
	is Lipschitz with constant $L$ as a function of $y$ for every $x$. 
	Let $\{(X_i,Y_i,\delta_i)\}_{i=1}^n$ be i.i.d.\ random vectors in $\cD \times \{0,1\}$. 
	Let $f_Y(y)$ be the density of $Y_i$ with respect to some measure $\nu$.
	Define
	\begin{align*}
		\cH_n(y) & = \frn \sumi \delta_i g(X_i,y) \\
		\cH(y) & = E[\delta_i g(X_i,y)] \\
		\Delta_n(y) & = \cH_n(y) - \cH(y)	 \\
		Z_n & = \frn \sumj \Delta_n(Y_j) - \int f_Y(y)  \Delta_n(y) \, d\nu(y)	
	\end{align*}
	Then $Z_n = o_P(n^{-1/2})$.
\end{proposition}

\noindent
\textit{Proof}: 
Write $h(x,\delta,y) = \delta g(x,y)$
From Proposition 2 above, the class $\mathscr{F} = \{ \delta h(\cdot,\cdot,y): y \in [a_2,b_2] \}$ 
is Donsker and therefore Glivenko-Cantelli. Therefore $\sup_{y \in [a_2,b_2]}
|\Delta_n(y)| \to 0$ almost surely.
Now let $\mathscr{C}$ denote the class of Lipschitz functions on $[a_2, b_2]$ with constant $L$.
By \citet[][Theorem 2]{gine1986}, $\mathscr{C}$ is Donsker.
It is easy to see that $\Delta_n \in \mathscr{C}$.
The claim now follows from \citet[][Lemma 19.24]{vdv1998}. \ding{110}

\vspace*{4mm}

\begin{proposition}
	On $t_1\in[t_{1min},\tau]$, $r\in[R_L,R_U]$,
	\begin{align*}
		\sup_{t_1,r}\big|\wh\Lambda_{T_2| T_1}(r| t_1)-\Lambda_{T_2| T_1}(r| t_1)\big|
		& =O_{a.s.}\!\left(n^{-(1-\nu)/2} (\log n)^{1/2}\right), \\[3mm]
		\sup_{t_1,r}\big|\wh S_{T_2| T_1}(r| t_1)-S_{T_2|T_1}(r | t_1)\big|
		& = O_{a.s.}\!\left(n^{-(1-\nu)/2} (\log n)^{1/2}\right).
	\end{align*}
\end{proposition}

\noindent
\textit{Proof}:
Consider the Kaplan-Meier version of the Beran conditional survival estimator 
\begin{equation}
	\widehat{S}_{T_2|T_1}^\pss{KM}(t_2|t_1) = \Prodi_{t_1}^{t_2} \left( 1 -\frac{(nh_n)^{-1} \sum_{m=1}^n \delta_{1m} \cK((V_{1m}-t_1)/h_n) dN_{2m}(u)}
	{(nh_n)^{-1} \sum_{m=1}^n \delta_{1m} \cK((V_{1m}-t_1)/h_n) Y_{2m}(u)} \right),
\end{equation}
where $\Prodi_{t_1}^{t_2}$ denotes the product integral.
From Proposition 4.3 of \cite{vankeil1999}, we have 
$$
\sup_{t_1,r}\big|\wh{S}_{T_2| T_1}^\pss{KM}(r| t_1)-S_{T_2| T_1}(r| t_1)\big|
= O_{a.s.}\!\left(\sqrt{\frac{\log n}{nh}}\right)
= O_{a.s.}\!\left(n^{-(1-\nu)/2} (\log n)^{1/2}\right).
$$
We have $\wh\Lambda_{T_2| T_1} = \mathscr{H}(\widehat{S}_{T_2|T_1}^\pss{KM})$,
where $\mathscr{H}$ is the map defined by
$$
\mathscr{H}({q})(t) = - \int_0^t q(s^-)^{-1} \, dq(s)
$$
Accordingly,
$$
|\wh\Lambda_{T_2| T_1}(r| t_1)-\Lambda_{T_2| T_1}(r| t_1)|
\leq (s_{min} - o_P(1))^{-1} |\wh{S}_{T_2| T_1}^\pss{KM}(r| t_1)-S_{T_2| T_1}(r| t_1)|.
$$
The result follows. \ding{110}

\underline{Main Proof of Theorem 1}

We now proceed with the proof of Theorem 1.

\vspace*{4mm}

We have
\begin{align*}
	\widehat{G}_1(t_1^*) - G_1(t_1^*) = \frac{1}{n} \sumi [\delta_{1i} I(V_{1i} \leq t_1^*) \widehat{Q}(V_{1i})
	- G_1(t_1^*)] = D_1 + D_2	
\end{align*}
with
\begin{align*}
	D_1 = \frac{1}{n} \sumi \left[ \delta_{1i} I(V_{1i} \leq t_1^*)Q(V_{1i}) - 	G_1(t_1^*) \right], \\
	D_2 = \frac{1}{n} \sumi \delta_{1i} I(V_{1i} \leq t_1^*) (\widehat{Q}(V_{1i}) - Q(V_{1i})).
\end{align*}
Here,
\begin{align*}
	& Q(v_1) = B(v_1)^{-1}, \\
	& \widehat{B}(v_1) = \frn \sum_{j=1}^n \widehat{\cA}(R_j,v_1) \sttwh(R_j|v_1), \\
	& \widehat{\cA}(r,v_1) = \sttwoh(r)^{-1} \wh{S}_W((v_1-r)_+).
\end{align*}
$D_1$ is the average of mean-zero i.i.d.\ terms.
$D_2$ can be written as $D_2 = -D_{21} + D_{22}$, where
\begin{align*}
	D_{21} & = \frn \sumi \delta_{1i} I(V_{1i} \leq t_1^*) B(V_{1i})^{-2}(\wh{B}(V_{1i}) - B(V_{1i})) \\	
	D_{22} & = \frac{1}{n} \sumi \delta_{1i} I(V_{1i} \leq t_1^*) \bar{B}_i^{-3}(\wh{B}(V_{1i}) - B(V_{1i}))^2 		
\end{align*}
where $\bar{B}_i$ is a value lying between $\wh{B}(V_{1i})$ and $B(V_{1i})$.
Define
$$
\widetilde{B}(v_1) = \frn \sum_{j=1}^n \cA(R_j,v_1) \sttw(R_j|v_1), \\
$$
We can then write
\begin{equation}
	\wh{B}(v_1) - B(v_1)
	= (\wh{B}(v_1) - \widetilde{B}(v_1)) + (\widetilde{B}(v_1) - B(v_1))
	\label{bee}
\end{equation}
By Donsker's theorem, under our assumed conditions, $\sup_{v_1} |\widetilde{B}(v_1) - B(v_1)| = O_P(n^{-1/2})$.
By standard results for the Kaplan-Meier estimator, $\sup_{r,v_1} |\widehat{\cA}(r,v_1) - \cA(r,v_1)| =
O_P(n^{-1/2})$. By Proposition~4,
$
\sup_{t_1,r}\big|\wh S_{T_2| T_1}(r| t_1)-S_{T_2|T_1}(r | t_1)\big| =
O_{a.s.}\!\left(n^{-(1-\nu)/2} (\log n)^{1/2}\right).
$
Thus, $\sup_{v_1} |\wh{B}(v_1) - \widetilde{B}(v_1)| = O_P\left(n^{-(1-\nu)/2} (\log n)^{1/2}\right)$.
Accordingly, $\sup_{v_1} |\wh{B}(v_1) - {B}(v_1)| = O_P\left(n^{-(1-\nu)/2} (\log n)^{1/2}\right)$.
This implies that $D_{22} = O_P\left(n^{-(1-\nu)} \log n \right)$. Under our assumption that
$\nu < \tfrac{1}{3}$, we see that $D_{22} = o_P(n^{-1/2})$.

We now proceed to analyze $D_{21}$.
We can write $D_{21} = D_{211} + D_{212}$ with
\begin{align*}
	D_{211} & = \frn \sumi \delta_{1i} I(V_{1i} \leq t_1^*) B(V_{1i})^{-2}(\wh{B}(V_{1i}) - \widetilde{B}(V_{1i})) \\
	D_{212} & = \frn \sumi \delta_{1i} I(V_{1i} \leq t_1^*) B(V_{1i})^{-2}(\widetilde{B}(V_{1i}) - B(V_{1i})) 	
\end{align*}
We consider $D_{212}$ first. We have
\begin{align*}
	D_{212} & = \frn \sumi \delta_{1i} I(V_{1i} \leq t_1^*) B(V_{1i})^{-2}
	\left[ \frn \sum_{j=1}^n \cA(R_j,V_{1i}) \sttw(R_j|V_{1i})) \right. \\
	& \hspace{20mm} -  \left.  \int_{R_L}^{R_U} f_{R|R \leq T_2}(r) \cA(r,V_{1i}) \sttw(r|V_{1i}) \, dr \right] \\
	& = \frn \sumj \left[ \frn \sumi \delta_{1i} I(V_{1i} \leq t_1^*) B(V_{1i})^{-2} \cA(R_j,V_{1i})
	\sttw(R_j|V_{1i}) \right] \\
	& \hspace{20mm} -  \int_{R_L}^{R_U} f_{R|R \leq T_2}(r)
	\left[ \sumi \delta_{1i} I(V_{1i} \leq t_1^*) B(V_{1i})^{-2} \cA(r,V_{1i}) \sttw(r|V_{1i}) \right] dr \\
	& = \frn \sumj \cH_n(R_j) - \int_{R_L}^{R_U} f_{R|R \leq T_2}(r) \cH_n(r) \, dr, 
\end{align*}
where
\begin{align*}
	\cH_n(r) & = \frn \sumi \delta_{1i} I(V_{1i} \leq t_1^*) B(V_{1i})^{-2} \cA(r,V_{1i}) \sttw(r|V_{1i}). 
\end{align*}
Define
$$
\cH(r) = E[\delta_{1i} I(V_{1i} \leq t_1^*) B(V_{1i})^{-2}  \cA(r,V_{1i}) \sttw(r|V_{1i}) ].
$$
We can then write
\begin{align*}
	D_{212} & = 	
	\frn \sumj (\cH(R_j) - E[\cH(R_j)]) \\
	& \hspace{20mm}+ \left[ \frn \sumj (\cH_n(R_j) - \cH(R_j)) 
	- \int_{R_L}^{R_U} f_{R|R \leq T_2}(r) (\cH_n(r) - \cH(r)) \, dr \right]
\end{align*}
The first term is the average of mean-zero i.i.d.\ terms, while the second term
is $\op(n^{-1/2})$ by Proposition 3. The Lipschitz condition of Proposition 3
is satisfied because of the differentiability conditions on
$S_{T_2}, S_W$, and $S_{T_2|T_1}$ and the assumption that $S_{T_2}(\tau)>0$.

We now examine $D_{211}$. We can write $D_{211}= D_{2111} + D_{2112}$ with
\begin{align*}
	D_{2111} & = \frn \sumi \delta_{1i} I(V_{1i} \leq t_1^*) B(V_{1i})^{-2}	
	\left[ \frn \sumj \sttwh(R_j|V_{1i}) (\widehat{\cA}(R_j,V_{1i}) - \cA(R_j,V_{1i})) \right], \\
	D_{2112} & = \frn \sumi \delta_{1i} I(V_{1i} \leq t_1^*) B(V_{1i})^{-2}	
	\left[ \frn \sumj \cA(R_j,V_{1i}) (\sttwh(R_j|V_{1i}) - \sttw(R_j|V_{1i})) \right].
\end{align*}

Let us analyze $D_{2111}$.
By Proposition 4,
\begin{equation*}
	\sup_{r \in [R_L,R_U], v_1 \in [t_{1min},\tau]} |\sttwh(r|v_1)-\sttw(r|v_1)| = \op(1),
\end{equation*}
so we will work with
$$
\tilde{D}_{2111} = \frn \sumi \delta_{1i} I(V_{1i} \leq t_1^*) B(V_{1i})^{-2}	
\left[ \frn \sumj \sttw(R_j|V_{1i}) (\widehat{\cA}(R_j,V_{1i}) - \cA(R_j,V_{1i})) \right], \\
$$
Recall that
\begin{align*}
	& \cA(r,v_1) = \sttwo(r)^{-1} S_W((v_1-r)_+), \\
	& \widehat{\cA}(r,v_1) = \sttwoh(r)^{-1} \wh{S}_W((v_1-r)_+).	
\end{align*}
We can write
\begin{align}
	\widehat{\cA}(r,v_1) - \cA(r,v_1)
	& = \sttwo(r)^{-1} (\wh{S}_W((v_1-r)_+) - S_W((v_1-r)_+)) \nonumber \\
	& \hspace*{10mm} - \, S_W((v_1-r)_+) \sttwo(r)^{-2} (\sttwoh(r) -  \sttwo(r)) \nonumber \\
	& \hspace*{10mm} + \, o\left( \max\{\sttwoh(r) -  \sttwo(r), \wh{S}_W((v_1-r)_+) - S_W((v_1-r)_+)\} \right).
	\label{aa}
\end{align}
The last term is negligible relative to the first two terms. 
Let us denote by $\tilde{D}_{2111a}$ the part of $\tilde{D}_{2111}$ involving the first 
term of (\ref{aa}) and by $\tilde{D}_{2111b}$ the part of $\tilde{D}_{2111}$ involving the 
second term of (\ref{aa}). In regard to $\tilde{D}_{2111a}$, we can write
$$
\tilde{D}_{2111a} = \frn \sumi \delta_{1i} I(V_{1i} \leq t_1^*) B(V_{1i})^{-2}	
\left[ \frn \sumj \sttw(R_j|V_{1i}) \sttwo(R_j)^{-1} (\wh{S}_W(\cD_{ij})
- S_W(\cD_{ij})) \right].
$$
with $\cD_{ij} = (V_{1i}-R_j)_+$.
Define $\cB_k = V_{2k} - R_k$, $N_{wk}(s) = (1-\delta_{2k})I(\cB_k \leq s)$, and $\cY_w(s) = P(\cB_k \geq s)$.
Recalling from Section 2.1 of the main paper that
\begin{equation}
	f_{T_{1},T_{2},R,W|R \leq T_{2}}(t_1,t_2,r,w) = a^{-1} f_{T_{1},T_{2}}(t_1,t_2) f_{R}(r) f_{W}(w)  I(t_1 \leq t_2)I (r \leq~t_2), 
	\label{ffull}
\end{equation}
and
\begin{equation*}
	f_{R}(r) = a f_{R|R \leq T_{2}}(r)S_{T_2}(r)^{-1},
\end{equation*}
we have
\begin{align}
	\cY_w(s) & = P(V_{2} - R \geq s) \nonumber \\
	& = P(T_2 - R \geq s, W \geq s) \nonumber \\
	& = P(T_2 - R \geq s) S_W(s) \nonumber \\
	& = S_W(s) \int_0^\tau \int_{R_L}^{R_U} I(t_2 \geq s + r)
	a^{-1} f_{T_2}(t_2)f_R(r) \, dt_2 \, dr \nonumber \\
	& = S_W(s) \int_0^\tau \int_{R_L}^{R_U} I(t_2 \geq s + r)
	f_{R|R\leq T_2}(r) \, dt_2 \, dr  
	\label{yw}
\end{align}
Classical theory for the Kaplan-Meier estimator (see, for example, \cite{lo1985}), says that
\begin{align*}
	\widehat{S}_W(w) - S_W(w)
	= - \, \frn \sum_{k=1}^n S_W(w) \int_0^w \cY_w(s)^{-1} dM_{wk}(s) + o_{a.s.}(n^{-3/4} (\log n)^{3/4})
\end{align*}
with $dM_{wk}(s) = dN_{wk}(s) - \lambda_W(s)  I(\cB_k \geq s) ds$, 
where $\lambda_W$ is the hazard function of $W$ and the $o_{a.s.}$ term is uniform in $w$.
We can thus write
\begin{align*}
	-\tilde{D}_{2111a} & = \frn \sumi \delta_{1i} I(V_{1i} \leq t_1^*) B(V_{1i})^{-2} \\	
	& \left\{ \frn \sumj \sttw(R_j|V_{1i}) \sttwo(R_j)^{-1} S_W(\cD_{ij}) \right. \\
	& \left. 
	\left[ \frn \sum_{k=1}^n \int_0^{\tau^*} 
	I(\cD_{ij} \geq s) \cY_w(s)^{-1} dM_{wk}(s) \right] \right\}
	+ \op(n^{-1/2}) \\
	& = \frn \sum_{k=1}^n \int_0^{\tau^*} \Omega_n^\pss{a}(s) \cY_w(s)^{-1} dM_{wk}(s) + \op(n^{-1/2}).
\end{align*}
where
$$
\Omega_n^\pss{a}(s) = \left(\frn\right)^2  \sumi \sumj
\delta_{1i} I(V_{1i} \leq t_1^*) B(V_{1i})^{-2} 
\sttw(R_j|V_{1i}) \sttwo(R_j)^{-1} S_W(V_{1i} - R_j) I(R_j \leq V_{1i} - s).
$$
Define
\begin{align*}
	\bar{\Omega}^\pss{a}(s)
	= \int_0^\tau & \int_{R_L}^{R_U} f_{V_1}(v_1) f_{R|R\leq T_2}(r) \\
	& \hspace*{8mm} E[\delta_1|V_1=v_1] I(v_1 \leq t_1^*) B(v_1)^{-2} \sttw(r|v_1) \sttwo(r)^{-1} \\
	& \hspace*{20mm} S_W(v_1-r) I(r \leq v_1-s) \, dv_1 \, dr
\end{align*}
and $\Upsilon_n(s) = \Omega_n^\pss{a}(s) - \bar{\Omega}^\pss{a}(s)$. 
We can then write
\begin{align}
	-\tilde{D}_{2111a} 
	& = \frn \sum_{k=1}^n \int_0^{\tau^*} 
	\bar{\Omega}^\pss{a}(s) \cY_w(s)^{-1} dM_{wk}(s) 
	+ \frn \sum_{k=1}^n \int_0^{\tau^*} 
	\Upsilon_n(s) \cY_w(s)^{-1} dM_{wk}(s).
	\label{ups}
\end{align}
The first term on the right is the average of mean-zero i.i.d.\ terms.
We will show that the second term on the right is $\op(n^{-1/2})$.
First, we argue that $\sup_s |\Upsilon_n(s)| = \Op(n^{-1/2})$.
This can be seen as follows. The class of functions $a_{v,s}(r) = I(r \leq v-s)$
as $v$ ranges over $[0,\tau]$ and $s$ ranges over $[0,\tau^*]$ is Donsker
by an argument similar to that of \citet[][Example 19.6]{vdv1998},
and, under Assumption A3, the class of functions $b_v(r) =\sttw(r|v) S_W(v-r)$  
as $v$ ranges over $[0,\tau]$ is Donsker by \citet[][Example 19.7]{vdv1998}.
Hence, by \citet[][Example 19.20]{vdv1998}, the class of functions
$c_{v,s}(r) = a_{v,s}(r)b_v(r)$ 
as $v$ ranges over $[0,\tau]$ and $s$ ranges over $[0,\tau^*]$ is Donsker.
Consequently,
\begin{align*}
	& \sup_{v_1,s} \Bigg|
	\frn \sumj \sttw(R_j|v_1) \sttwo(R_j)^{-1} S_W(v_1 - R_j) I(R_j \leq v_1 - s) \\
	& \hspace*{10mm} - \int_{L_R}^{U_R} \sttw(r|v_1) \sttwo(r)^{-1} S_W(v_1 - r) I(r \leq v_1 - s)
	f_{R|R\leq T_2}(r) \, dr \Bigg| = \Op(n^{-1/2}).
\end{align*}
Similarly, the class $d_{r,s}(v_1) = I(v_1 \geq r + s)$
as $r$ ranges over $[R_L, R_U]$ and $s$ ranges over $[0,\tau^*]$ is Donsker, and therefore
\begin{align*}
	& \sup_s \Bigg| 
	\frn \sumi 
	\delta_{1i} I(V_{1i} \leq t_1^*) B(V_{1i})^{-2} \\
	& \hspace*{10mm} \int_0^\tau \sttw(r|V_{1i}) \sttwo(r)^{-1} S_W(V_{1i} - r) I(r \leq V_{1i} - s) 
	f_{R|R\leq T_2}(r) \, dv_1 \,  dr - \bar{\Omega}(s) \Bigg| = \Op(n^{-1/2}).	
\end{align*}
Note here that, by Assumptions A4 and A6,  $B(\cdot)^{-2}$ and $\sttwo(\cdot)^{-1}$ are uniformly bounded, so that the mappings
$r \mapsto \sttw(r|v_1) \sttwo(r)^{-1} S_W(v_1 - r) I(r \leq v_1 - s)$
and $v_1 \mapsto I(v_1 \leq t_1^*) B(v_1)^{-2}$
admit an integrable envelope, which justifies the application of Donsker class theory.
We have thus established that $\sup_s |\Upsilon_n(s)| = \Op(n^{-1/2})$.

We turn now to the second term on the right side of (\ref{ups}).
This term can be written as $J_1 - J_2$, where
\begin{align*}
	J_1 & = \frn \sum_{k=1}^n \int_0^{\tau^*} \Upsilon_n(s) \cY_w(s)^{-1} dN_{wk}(s) \\
	& = \frn \sum_{k=1}^n \Upsilon_n(W_k) \cY_w(W_k)^{-1} I(W_k \leq \tau^*) I(T_{2k} \geq R_k + W_k) \\
	J_2 & = \int_0^{\tau^*} \Upsilon_n(s) \cY_w(s)^{-1} \left( \frn \sum_{k=1}^n Y_{wk}(s) \right) 
	\lambda_W(s) \, ds
\end{align*}
The term $J_2$, in turn, can be written as $J_2 = J_{21} + J_{22}$ with
\begin{align*}
	J_{21} & = 	\int_0^{\tau^*} \Upsilon_n(s) \cY_w(s)^{-1} \cY_w(s) \lambda_W(s) \, ds \\
	& = \int_0^{\tau^*} \Upsilon_n(s) \cY_w(s)^{-1} 
	\left[ 
	S_W(s) \int_0^\tau \int_{R_L}^{R_U} I(t_2 \geq s + r)
	f_{R|R\leq T_2}(r) \, dt_2 \, dr  
	\right] \lambda_W(s) \, ds \\
	& = \int_0^\tau \int_{R_L}^{R_U} \int_0^{\tau} \Upsilon_n(s) \cY_w(s)^{-1} 
	I(s \leq t^*)I(t_2 \geq s + r) f_{T_2}(t_2)f_{R|R\leq T_2} f_W(s) \, dt_2 \, dr \, ds \\
	J_{22} & = \int_0^{\tau^*} \Upsilon_n(s) \cY_w(s)^{-1} \left\{ \left( \frn \sum_{k=1}^n Y_{wk}(s) \right) 
	- \cY_w(s) \right\} \, ds,
\end{align*}
where in the development for $J_{21}$ we have used (\ref{yw}).

Given that $\sup_s |\Upsilon_n(s)| = \Op(n^{-1/2})$,
the quantity $J_1 - J_{21}$ is $\op(n^{-1/2})$ by a straightforward extension of Proposition 3,
while $J_{22}$ is $\Op(n^{-1})$ since $\sup_s |\Upsilon_n(s)| = \Op(n^{-1/2})$ and the
term in braces in the integral is $\Op(n^{-1/2})$.
We have thus shown that the second term of (\ref{ups}) is $\op(n^{-1/2})$.
Thus,
\begin{align}
	-\tilde{D}_{2111a} 
	& = \frn \sum_{k=1}^n \int_0^{\tau^*} \bar{\Omega}^\pss{a}(s) \cY_w(s)^{-1} dM_{wk}(s) + \op(n^{-1/2})
\end{align}
This completes the analysis of $\tilde{D}_{2111a}$.

In a similar way, $-\tilde{D}_{2111b}$ can be written as
$$
-\tilde{D}_{2111b} = \frn \sum_{k=1}^n \int_0^{\tau} \bar{\Omega}^\pss{b}(s) \cY_2(s)^{-1} dM_{2k}(s) + \op(n^{-1/2})
$$ 
with $dM_{2k}(s) = dN_{2k}(s) - \lambda_{T_2}(s)  I(V_{2k} \geq s) ds$, $\cY_2(s) = P(V_{2k} \geq s)$, and
\begin{align*}
	\bar{\Omega}^\pss{b}(s)
	= \int_0^\tau & \int_{R_L}^{R_U} f_{V_1}(v_1) f_{R|R\leq T_2}(r) \\
	& \hspace*{8mm} E[\delta_1|V_1=v_1] I(v_1 \leq t_1^*) B(v_1)^{-2} \sttw(r|v_1) \sttwo(r)^{-1} \\
	& \hspace*{20mm} S_W(v_1-r) I(r \geq s) \, dv_1 \, dr
\end{align*}
Finally, let us analyze $D_{2112}$. Recall that
$$
D_{2112} = \frn \sumi \delta_{1i} I(V_{1i} \leq t_1^*) B(V_{1i})^{-2}	
\left[ \frn \sumj \cA(R_j,V_{1i}) (\sttwh(R_j|V_{1i}) - \sttw(R_j|V_{1i})) \right].
$$
We have
\begin{align*}
	\sttwh(R_j|V_{1i}) - \sttw(R_j|V_{1i})	
	& = \exp(-\widehat{\Lambda}_{T_2|T_1}(R_j|V_{1i})) 
	-  \exp(-{\Lambda}_{T_2|T_1}(R_j|V_{1i})) \\
	& = -\widetilde{S}_{T_2|T_1}(R_j|V_{1i})
	(\widehat{\Lambda}_{T_2|T_1}(R_j|V_{1i}) - {\Lambda}_{T_2|T_1}(R_j|t_1)) 
\end{align*}
with $\widetilde{S}_{T_2|T_1}(t_2|t_1) = \exp(-\widetilde{\Lambda}_{T_2|T_1}(t_2|t_1))$,
where $\widetilde{\Lambda}_{T_2|T_1}(t_2|t_1)$ lies between $\widehat{\Lambda}_{T_2|T_1}(t_2|t_1)$ and
$\Lambda_{T_2|T_1}(t_2|t_1)$.

Now, using (\ref{ffull}), we have
$$
f_{T_{1},R,\delta_1|R \leq T_{2}}(t_1,r,1) = a^{-1} f_{T_{1}}(t_1) f_{R}(r) S_W((t_1-r)_+) \int_{\max\{r,t_1\}}^\infty f_{T_2|T_1}(t_2|t_1) \, dt_2
$$
Thus, the conditional density of $T_2$ conditional on $R \leq T_2, T_1=t_1, R=r, \delta_1=1$ is
$$
f^*(t_2|t_1,r) =  a^*(t_1,r)^{-1} f_{T_2|T_1}(t_2|t_1)I(t_2 \geq t_1), \quad a^*(t_1,r) = \int_{\max\{r,t_1\}}^\infty f_{T_2|T_1}(\tilde{t}_2|t_1) \, d\tilde{t}_2,
$$
the corresponding survival function is, for $t_2 \geq t_1$,
$$
S^*(t_2|t_1,r) = a^*(t_1,r)^{-1} S_{T_2|T_1}(t_2|t_1),
$$
and the corresponding hazard function is, for $t_2 \geq t_1$,
$$
\lambda^*(t_2|t_1,r) = \frac{f^*(t_2|t_1,r)}{S^*(t_2|t_1,r)} = \frac{f_{T_2|T_1}(t_2|t_1)}{S_{T_2|T_1}(t_2|t_1)} = \lambda_{T_2|T_1}(t_2|t_1)
$$
Accordingly, defining
$$
M_{2m}^*(s) = \delta_{1m} \left[ N_{2m}(s) - \int_0^s \lambda_{T_2|T_1}(u|V_{1m})
Y_{2m}(u) \, du \right],
$$ 
we see that $M_{2m}^*(s)$ is a martingale with respect to the filtration
$$
\cF_s = \sigma(R_i, V_{1i}, \delta_{1i}, i = 1, \ldots, n,
N_{2m}(u), Y_{2m}(u), u \in [0,s], m = 1, \ldots, n)
$$
(note that $\delta_{1m}$ and $V_{1m}$ are $\cF_0$-measurable).
Define further
\begin{align*}
	\cW_m(t_1,u) & = \delta_{1m} \cK\left(\frac{V_{1m}-t_1}{h}\right)Y_{2m}(u) \\
	\dot{\lambda}_{T_2|T_1}(u|t_1) & = \frac{\partial}{\partial t_1} \, {\lambda}_{T_2|T_1}(u|t_1) \\
	\ddot{\lambda}_{T_2|T_1}(u|t_1) & = \frac{\partial^2}{\partial t_1^2} \, {\lambda}_{T_2|T_1}(u|t_1)  
\end{align*}
We can then write
\begin{align*}
	& \widehat{\Lambda}_{T_2|T_1}(t_2|t_1) - \Lambda_{T_2|T_1}(t_2|t_1)	\\
	& \hspace*{0mm} = \frnh \summ \int_{t_1}^{t_2} A_0(t_1,u)^{-1} \cW_m(t_1,u) dN_{2m}(u)
	- \, \int_{t_1}^{t_2} 	\lambda_{T_2|T_1}(u|t_1) \, du \\
	& \hspace*{5mm} = \frnh \summ \int_{t_1}^{t_2} A_0(t_1,u)^{-1} \cW_m(t_1,u) \lambda_{T_2|T_1}(u|V_{1m}) \, du \\
	& \hspace*{10mm} + \, \frnh \summ \int_{t_1}^{t_2} A_0(t_1,u)^{-1} \cW_m(t_1,u) \, dM_{2m}^*(u) \\
	& \hspace*{10mm} - \, \int_{t_1}^{t_2} 	\lambda_{T_2|T_1}(u|t_1) \, du \\
	& \hspace*{5mm} = \frnh \summ \int_{t_1}^{t_2} A_0(t_1,u)^{-1} \cW_m(t_1,u)
	\lambda_{T_2|T_1}(u|t_1) \, du \\
	& \hspace*{10mm} + \, \frnh \summ \int_{t_1}^{t_2} A_0(t_1,u)^{-1} \cW_m(t_1,u)
	\dot{\lambda}_{T_2|T_1}(u|t_1)(V_{1m} - t_1) \, du \\
	& \hspace*{10mm} + \, \frac{1}{2} \frnh \summ \int_{t_1}^{t_2} A_0(t_1,u)^{-1} \cW_m(t_1,u)
	\ddot{\lambda}_{T_2|T_1}(u|\widetilde{t}_{1m}(u))(V_{1m} - t_1)^2 \, du \\
	& \hspace*{10mm} + \, \frnh \summ \int_{t_1}^{t_2} A_0(t_1,u)^{-1} \cW_m(t_1,u) \, dM_{2m}^*(u) \\
	& \hspace*{10mm} - \, \int_{t_1}^{t_2} 	\lambda_{T_2|T_1}(u|t_1) \, du \\
	& \mbox{(where $\widetilde{t}_{1m}(u)$ lies between $t_1$ and $V_{1m}$)} \\
	& \hspace*{5mm} = \int_{t_1}^{t_2} A_0(t_1,u)^{-1} \dot{\lambda}_{T_2|T_1}(u|t_1)
	\left[ \frnh \summ \cW_m(t_1,u)(V_{1m}-t_1) \right] \, du \\
	& \hspace*{10mm} + \, \int_{t_1}^{t_2} A_0(t_1,u)^{-1} 
	\left[ \frac{1}{2} \frnh \summ \ddot{\lambda}_{T_2|T_1}(u|\widetilde{t}_{1m}(u))
	\cW_m(t_1,u)(V_{1m} - t_1)^2 \right] \, du \\
	& \hspace*{10mm} + \, \frnh \summ \int_{t_1}^{t_2} A_0(t_1,u)^{-1} \cW_m(t_1,u) \, dM_{2m}^*(u) \\
	& \hspace*{5mm} = \cT_1(t_1,t_2) + \cT_2(t_1,t_2) + \cT_3(t_1,t_2)
\end{align*}
with
\begin{align*}
	\cT_1(t_1,t_2) & = h \int_{t_1}^{t_2} A_0(t_1,u)^{-1} \dot{\lambda}_{T_2|T_1}(u|t_1) A_1(t_1,u) \, du, \\
	\cT_2(t_1,t_2) & = \frac{1}{2} \, h^2 \int_{t_1}^{t_2} A_0(t_1,u)^{-1} 	
	\left[ \frnh \summ \ddot{\lambda}_{T_2|T_1}(u|\widetilde{t}_{1m}(u))
	\cW_m(t_1,u) \left(\frac{V_{1m} - t_1}{h}\right)^2 \right] \, du, \\
	\cT_3(t_1,t_2) & = \frnh \summ \int_{t_1}^{t_2} A_0(t_1,u)^{-1} \cW_m(t_1,u) \, dM_{2m}^*(u). 	
\end{align*}
Using Proposition 1, we get
\begin{align*}
	\cT_1(t_1,t_2) & = h \Op(h) = \Op(h^2), \\
	|\cT_2(t_1,t_2)| & \leq \tfrac{1}{2} h^2 [\min_u A_0(t_1,u)]^{-1} \|\ddot{\lambda}_{T_2|T_1}\|_\infty 
	\int_0^{t_2} A_2(t_1,u) \, du = \Op(h^2).
\end{align*}
We have 
$$
D_{2112} = E_1 + E_2 + E_3 
$$
with
\begin{align*}
	E_1 & = - \, \frn \sumi \delta_{1i} I(V_{1i} \leq t_1^*) B(V_{1i})^{-2}	
	\left[ \frn \sumj \cA(R_j,V_{1i}) \sttwt(R_j|V_{i1}) \cT_1(V_{1i},R_j) \right] \\
	E_2 & = - \, \frn \sumi \delta_{1i} I(V_{1i} \leq t_1^*) B(V_{1i})^{-2}	
	\left[ \frn \sumj \cA(R_j,V_{1i}) \sttwt(R_j|V_{i1}) \cT_2(V_{1i},R_j) \right] \\
	E_3 & = - \, \frn \sumi \delta_{1i} I(V_{1i} \leq t_1^*) B(V_{1i})^{-2}	
	\left[ \frn \sumj \cA(R_j,V_{1i}) \sttwt(R_j|V_{i1}) \cT_3(V_{1i},R_j) \right] \\
\end{align*}
We saw that $\cT_1(t_1,t_2) = \Op(h^2)$ and $\cT_2(t_1,t_2) = \Op(h^2)$, so that $E_1 = \Op(h^2)$
and $E_2 = \Op(h^2)$.
Since, by assumption, $h_n = O(n^{-\nu})$ with $\nu > \tfrac{1}{4}$, we get
$\sqrt{n} \, E_1 \cp 0$ and $\sqrt{n} \, E_2 \cp 0$.	

We now examine $E_3$. We have
\begin{align*}
	E_3 & = - \, \frn \sumi \delta_{1i} I(V_{1i} \leq t_1^*) B(V_{1i})^{-2}
	\left\{ \frn \sumj \cA(R_j,V_{1i}) \sttwt(R_j|V_{1i}) \vphantom{\int_0^{R_j}} \right. \\
	& \hspace*{30mm} \left. \left[ \frnh \summ \int_0^{R_j} I(u\geq V_{1i}) A_0(V_{1i},u)^{-1} \cW_m(V_{1m},u) \, 
	dM_{2m}^*(u) \right] \right\} \\
	& = - \, \frn \summ \int_0^\tau \frnh \sumi \delta_{1i} I(V_{1i} \leq t_1^*) B(V_{1i})^{-2}
	A_0(V_{1i},u)^{-1} \\
	& \hspace*{30mm}\left[ \frn \sumj \cA(R_j,V_{1i})  \sttwt(R_j|V_{1i}) I(R_j \geq u) \right] \\
	& \hspace*{30mm} \left[ \delta_{1m} \cK\left(\frac{V_{1m}-V_{1i}}{h}\right) Y_{2m}(u) \right]
	dM_{2m}^*(u) \\
	& = - \, \frn \summ \delta_{1m} \int_0^\tau
	\frn \sumi \delta_{1i} \cK\left(\frac{V_{1i}-V_{1m}}{h}\right)
	I(V_{1i} \leq t_1^*) B(V_{1i})^{-2} A_0(V_{1i},u)^{-1} \\
	& \hspace*{20mm} \left[ \frn \sumj \cA(R_j,V_{1i})  \sttwt(R_j|V_{1i}) I(R_j \geq u) \right]
	Y_{2m}(u) dM_{2m}^*(u) \\
	& = - \, \frn \summ \delta_{1m} \int_0^\tau \cI(u,V_{1m}) Y_{2m}(u) dM_{2m}^*(u)
\end{align*}                                              
with
\begin{align*}
	\cI(u,v) & = \frnh \sumi I(V_{1i} \leq t_1^*) B(V_{1i})^{-2} A_0(V_{1i},u)^{-1}
	\cK\left(\frac{V_{1i}-v}{h}\right) \\
	& \hspace*{20mm} \left[ \frn \sumj \cA(R_j,V_{1i})  \sttwt(R_j|V_{1i}) I(R_j \geq u) \right].
\end{align*}
Define
$$
\Phi(u,v) = E[\cA(R_j,v)  \sttw(R_j|v) I(R_j \geq u) ]
$$
We then have
\begin{align*}
	\cI(u,v) & = \frnh \sumi \delta_{1i} \cK\left(\frac{V_{1i}-v}{h}\right)
	I(V_{1i} \leq t_1^*) B(V_{1i})^{-2} \Phi(u,V_{1i}) a_0(V_{1i},u)^{-1} + \op(1) \\
	& = \frnh \sumi \delta_{1i} \cK\left(\frac{V_{1i}-v}{h}\right) \cG(u,V_{1i}) + \op(1),
\end{align*}
where 
$$
\cG(u,v) = I(v \leq t_1^*) B(v)^{-2} \Phi(u,v)a_0(v,u)^{-1}.
$$
Define
$$
\cI^\circ(u,v) = P(\delta_{1i}=1|T_{2i} \geq R_i) \kappa(v)\cG(u,v),
$$
where $\kappa(v)$ is the conditional density of $V_{1i}$ given $T_{2i} \geq R_i$ and $\delta_{1i}=1$.  
Then, by arguments similar to those used to prove Proposition 1, we have
\begin{equation}
	\sup_{u,v} |\cI(u,v) - \cI^\circ(u,v)| \cp 0.
	\label{ii}
\end{equation}
This gives us $E_3 = -(E_{31}+E_{32}(\tau))$, where
\begin{align*}
	E_{31} & = \frn \summ \int_0^{\tau} \cI^\circ(u,V_{1m}) Y_{2m}(u) dM_{2m}^*(u),  \\	
	E_{32}(x) & = \frn \summ \int_0^x (\cI(u,V_{1m})-\cI^\circ(u,V_{1m})) Y_{2m}(u) dM_{2m}^*(u)		
\end{align*}
$E_{31}$ is the average of i.i.d.\ mean-zero terms.
By Theorem II.3.1 of \cite{andersen1993}, $\sqrt{n} \, E_{32}(x)$ is a martingale
with predictable variation given by
$$
\langle \sqrt{n} \, E_{32}(x), \sqrt{n} \, E_{32}(x) \rangle
= \frn \summ \int_0^x (\cI(u,V_{1m})-\cI^\circ(u,V_{1m}))^2 Y_{2m}(u) \lambda_2(u|V_{1m}) \, du,
$$	
which converges in probability to 0 as $n \to \infty$ by (\ref{ii}) and the fact that $Y_{2m}$ and $\lambda_2$ are bounded.
Given this result, in combination with Lenglart's inequality as given in Display (2.5.18)
of \cite{andersen1993}, we conclude that $\sqrt{n} \, E_{32}(\tau) \cp 0$.

In conclusion, we have 
$$
\widehat{G}_1(t_1^*) - {G}_1(t_1^*) = \frac{1}{n} \sumi U_i(t_1^*) + o_{P}(n^{-1/2}),
$$
where
$$
U_i(t_1^*) =  \sum_{k=1}^5 U_{ik}(t_1^*)
$$
with
\begin{align*}
	U_{i1}(t_1^*) & = \delta_{1i} I(V_{1i} \leq t_1^*)Q(V_{1i}) - G_1(t_1^*) \\
	U_{i2}(t_1^*) & = \cH(R_i) - E[\cH(R_i)] \\
	U_{i3}(t_1^*) & = - \int_0^{\tau^*} \bar{\Omega}^\pss{a}(s) \cY_w(s)^{-1} dM_{wi}(s) \\
	U_{i4}(t_1^*) & = \int_0^{\tau^*} \bar{\Omega}^\pss{b}(s) \cY_2(s)^{-1} dM_{2i}(s) \\
	U_{i5}(t_1^*) & = - \int_0^{\tau} \cI^\circ(u,V_{1i}) Y_{2i}(u) dM_{2i}^*(u)  
\end{align*}
All of the terms $U_{ik}(t_1^*)$ have expectation zero, and none of them involves the bandwidth $h$.

\newpage

\section{Additional Simulation Results for Different Sample Sizes}

The following pages present simulation results for $n=2,500$ and $n=7,500$.

\begin{figure}[H]
	\centering
	\includegraphics[width=1\textwidth, height=0.85\textheight]{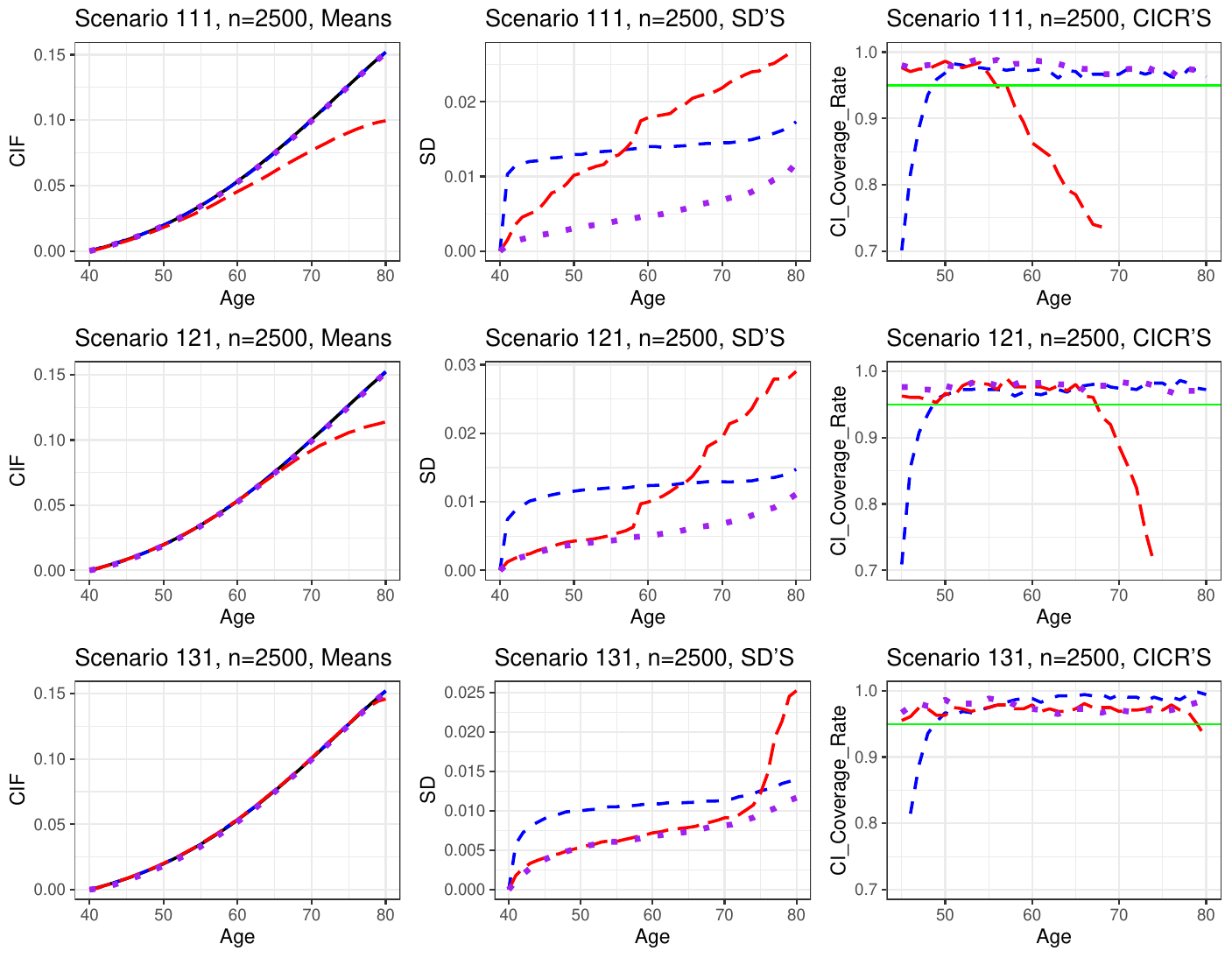}
	\caption{Simulation results for configurations 111, 121, and 131: Mean over estimates, standard deviations (SD), and 
		empirical coverage rates of 95\% pointwise confidence intervals, for the Aalen-Johansen (AJ), Gorfine-Zucker-Shoham (GZS),
		and new estimators. Solid black line – true curve, blue dashed line – AJ estimator, red long-dashed line - GZS estimator, 
		purple dotted line – new estimator. }\label{fig:f1} 
\end{figure}

\begin{figure}[H]
	\centering
	\includegraphics[width=1\textwidth, height=0.85\textheight]{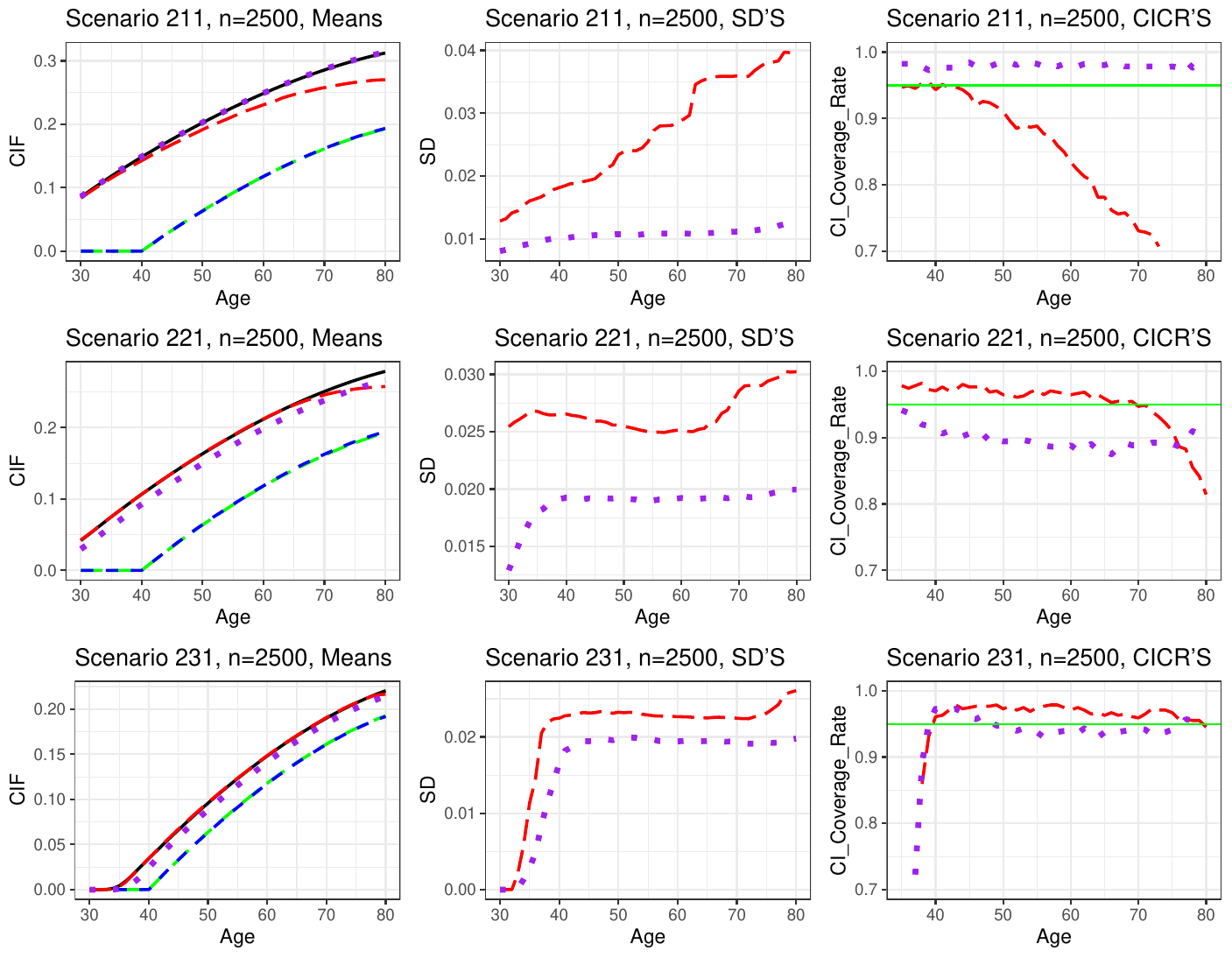}
	\caption{Simulation results for configurations 211, 221, and 231: Mean over estimates, standard deviations (SD), and 
		empirical coverage rates of 95\% pointwise confidence intervals, for the Aalen-Johansen (AJ), Gorfine-Zucker-Shoham (GZS),
		and new estimators. Solid black line – true curve, blue dashed line – AJ estimator, red long-dashed line - GZS estimator, 
		purple dotted line – new estimator. }\label{fig:f2} 
\end{figure}

\begin{figure}[H]
	\centering
	\includegraphics[width=1\textwidth, height=0.85\textheight]{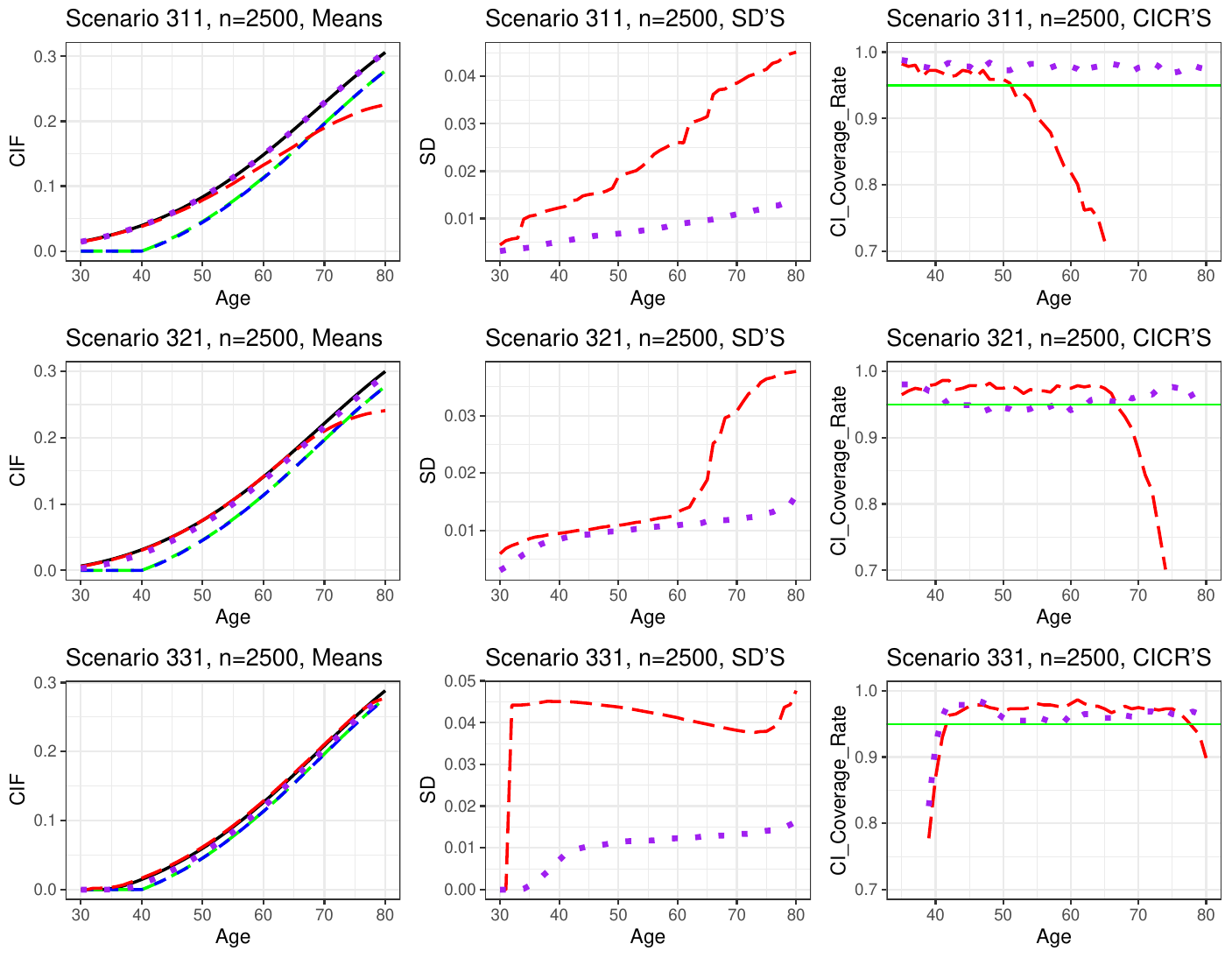}
	\caption{Simulation results for configurations 311, 321, and 331: Mean over estimates, standard deviations (SD), and 
		empirical coverage rates of 95\% pointwise confidence intervals, for the Aalen-Johansen (AJ), Gorfine-Zucker-Shoham (GZS),
		and new estimators. Solid black line – true curve, blue dashed line – AJ estimator, red long-dashed line - GZS estimator, 
		purple dotted line – new estimator. }\label{fig:f3} 
\end{figure}

\begin{figure}[H]
	\centering
	\includegraphics[width=1\textwidth, height=0.85\textheight]{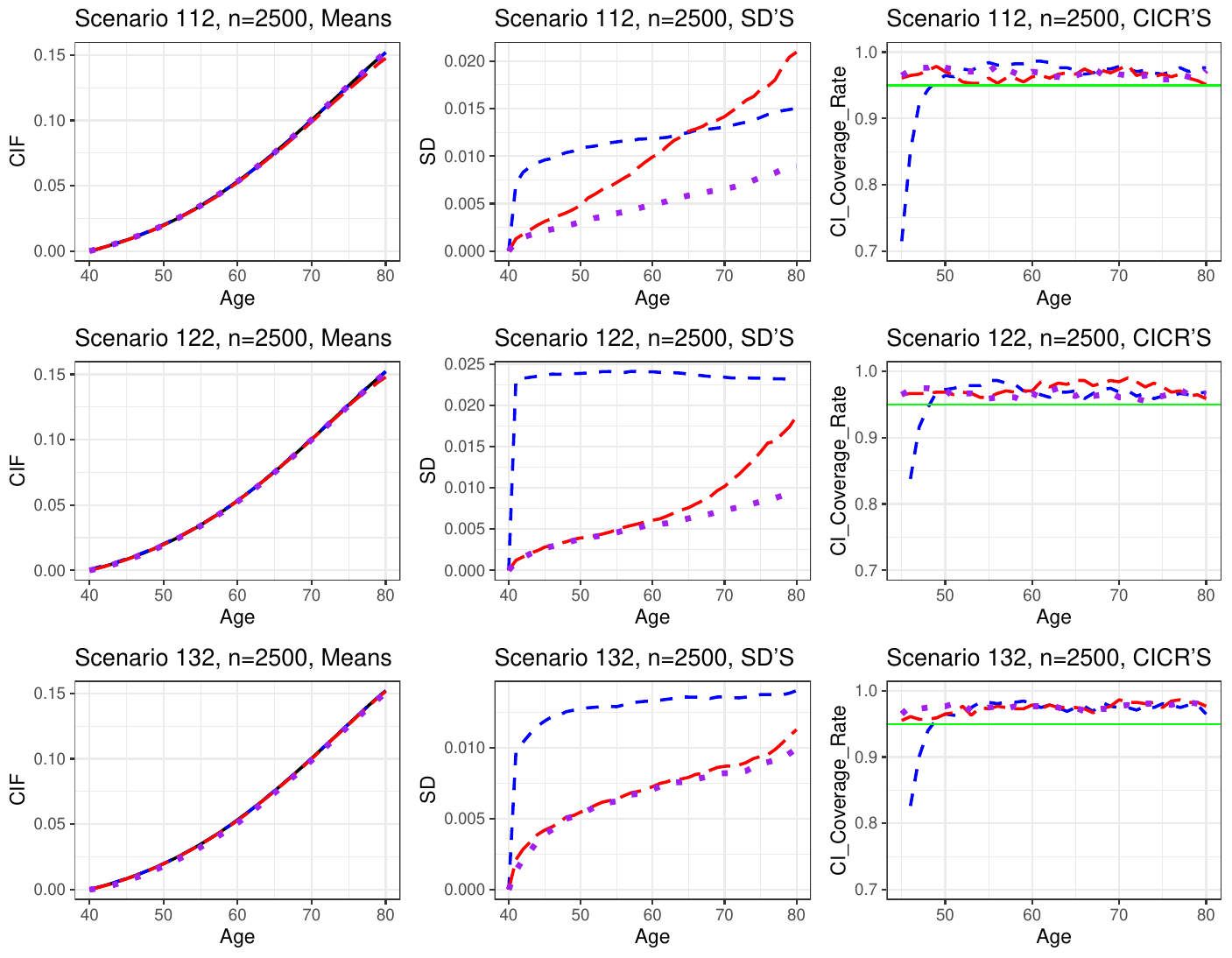}
	\caption{Simulation results for configurations 112, 122, and 132: Mean over estimates, standard deviations (SD), and 
		empirical coverage rates of 95\% pointwise confidence intervals, for the Aalen-Johansen (AJ), Gorfine-Zucker-Shoham (GZS),
		and new estimators. Solid black line – true curve, blue dashed line – AJ estimator, red long-dashed line - GZS estimator, 
		purple dotted line – new estimator. }\label{fig:f4} 
\end{figure}

\begin{figure}[H]
	\centering
	\includegraphics[width=1\textwidth, height=0.85\textheight]{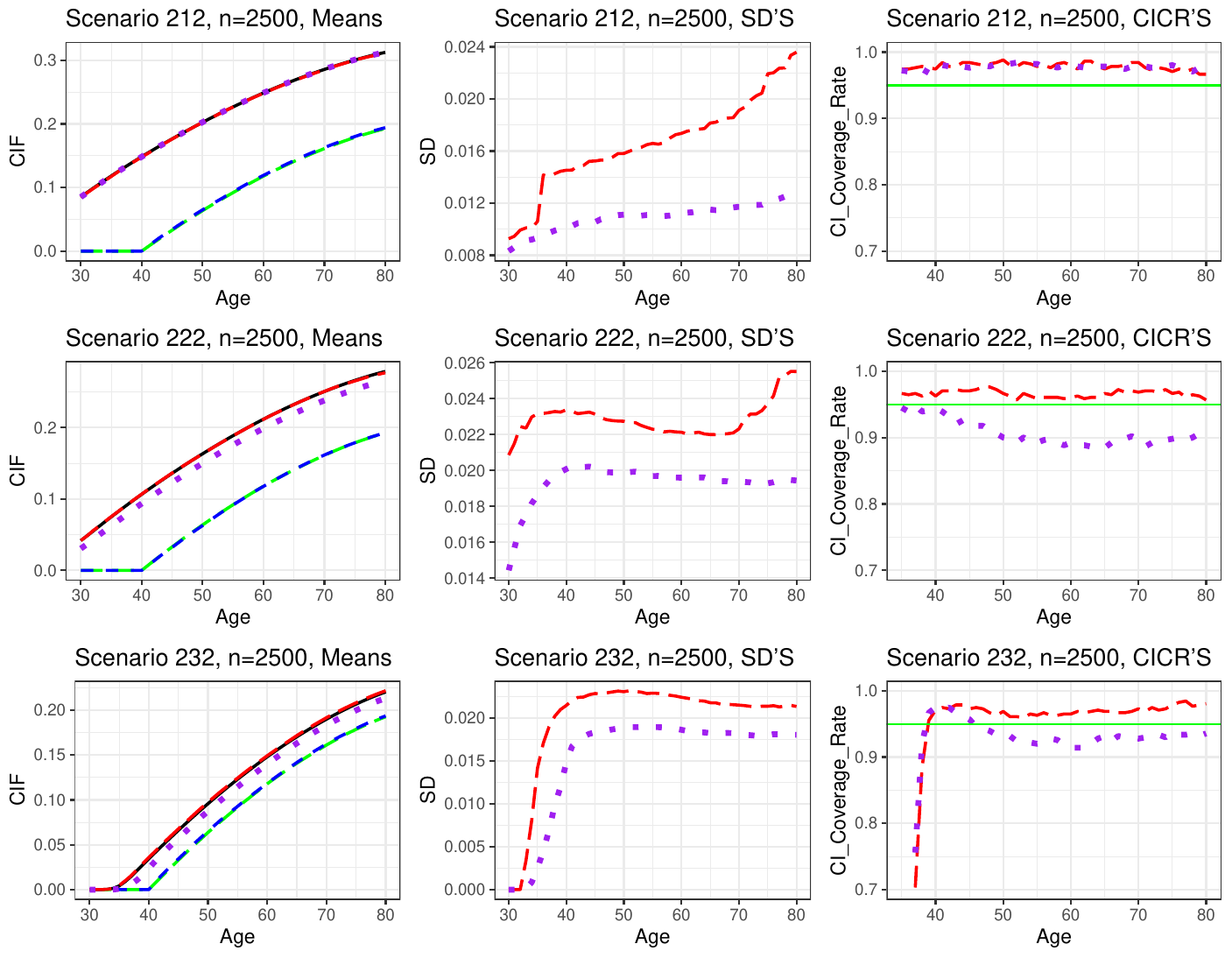}
	\caption{Simulation results for configurations 212, 222, and 232: Mean over estimates, standard deviations (SD), and 
		empirical coverage rates of 95\% pointwise confidence intervals, for the Aalen-Johansen (AJ), Gorfine-Zucker-Shoham (GZS),
		and new estimators. Solid black line – true curve, blue dashed line – AJ estimator, red long-dashed line - GZS estimator, 
		purple dotted line – new estimator. }\label{fig:f5} 
\end{figure}

\begin{figure}[H]
	\centering
	\includegraphics[width=1\textwidth, height=0.85\textheight]{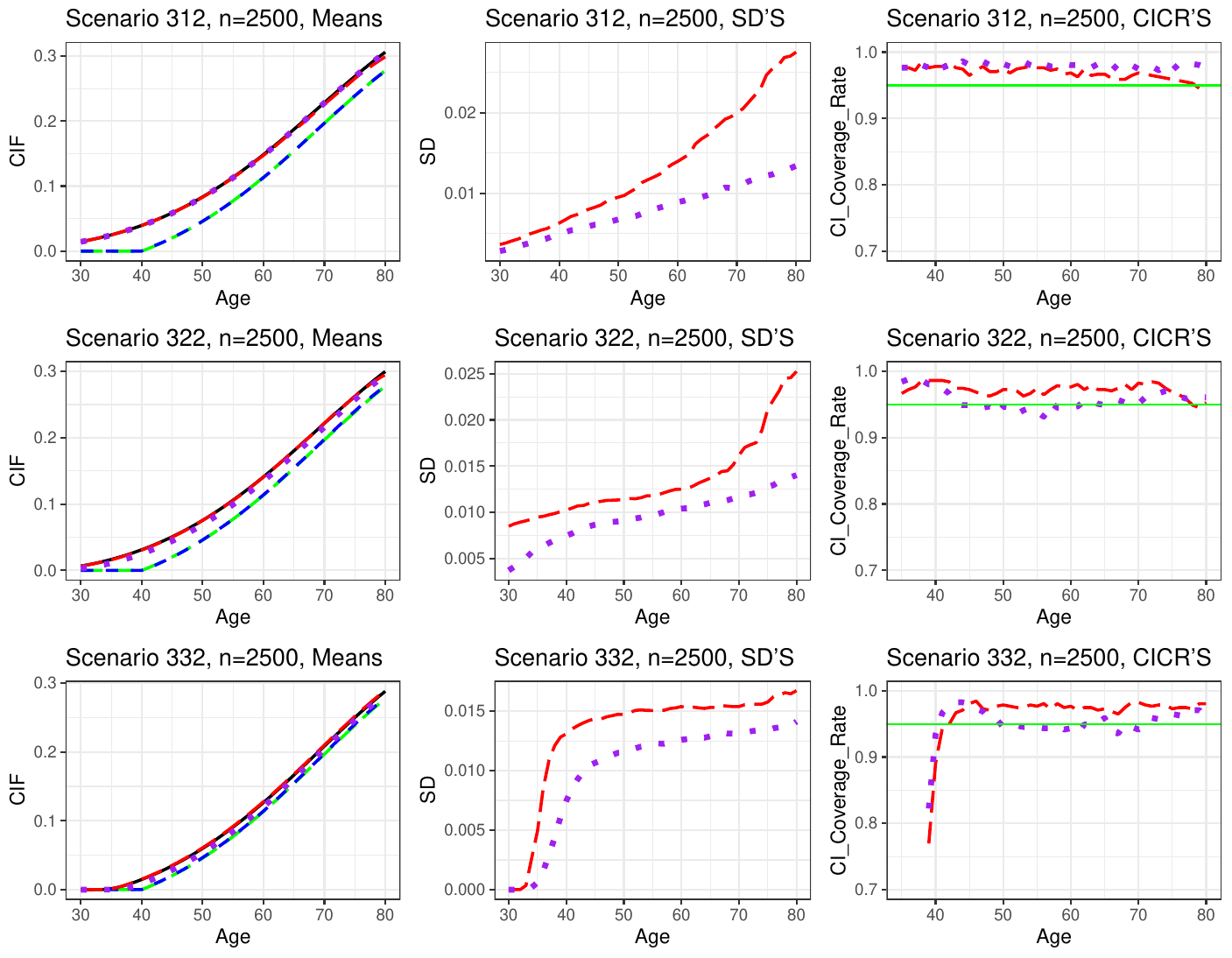}
	\caption{Simulation results for configurations 312, 322, and 332: Mean over estimates, standard deviations (SD), and 
		empirical coverage rates of 95\% pointwise confidence intervals, for the Aalen-Johansen (AJ), Gorfine-Zucker-Shoham (GZS),
		and new estimators. Solid black line – true curve, blue dashed line – AJ estimator, red long-dashed line - GZS estimator, 
		purple dotted line – new estimator. }\label{fig:f6} 
\end{figure}

\begin{figure}[H]
	\centering
	\includegraphics[width=1\textwidth, height=0.85\textheight]{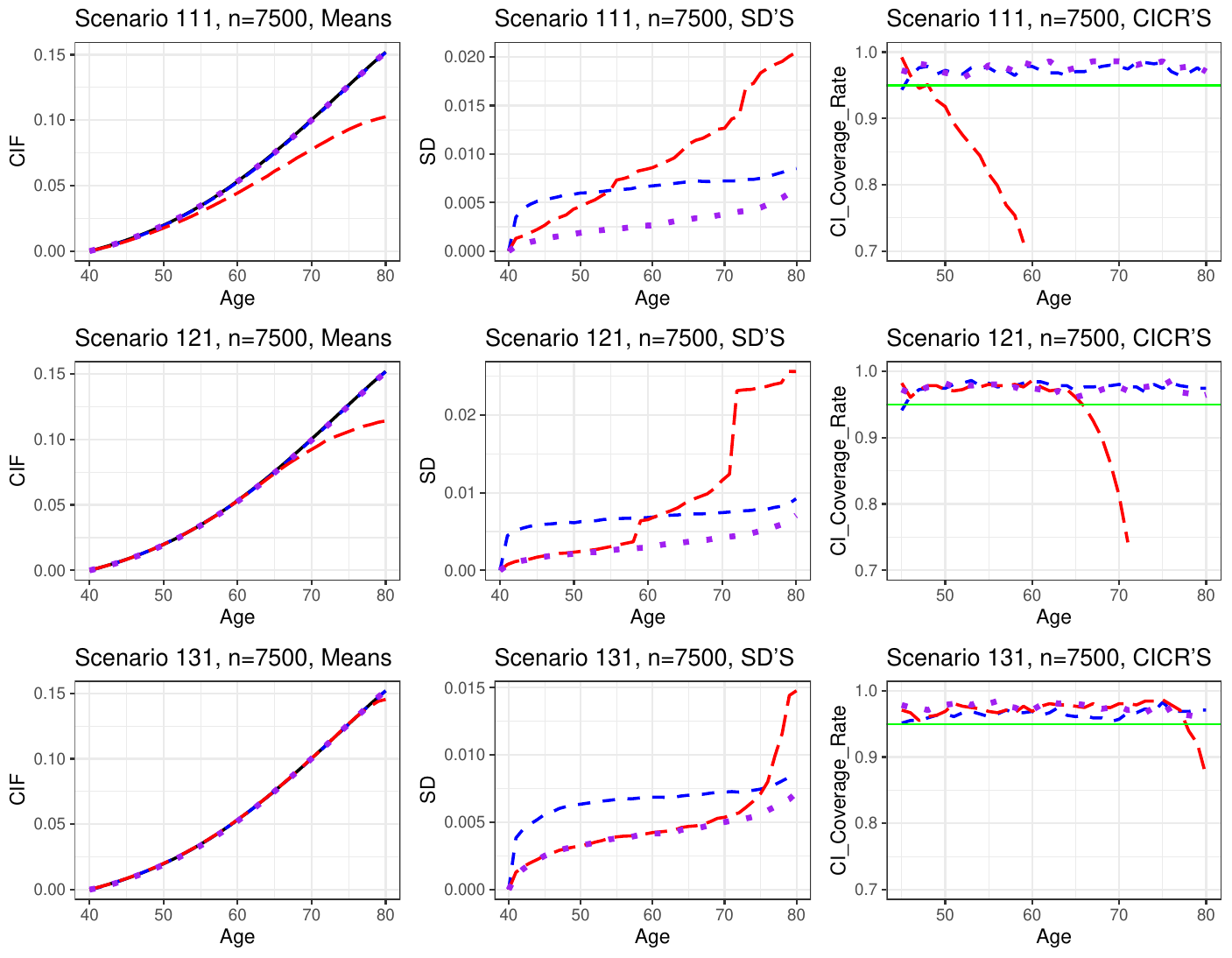}
	\caption{Simulation results for configurations 111, 121, and 131: Mean over estimates, standard deviations (SD), and 
		empirical coverage rates of 95\% pointwise confidence intervals, for the Aalen-Johansen (AJ), Gorfine-Zucker-Shoham (GZS),
		and new estimators. Solid black line – true curve, blue dashed line – AJ estimator, red long-dashed line - GZS estimator, 
		purple dotted line – new estimator. }\label{fig:f7} 
\end{figure}

\begin{figure}[H]
	\centering
	\includegraphics[width=1\textwidth, height=0.85\textheight]{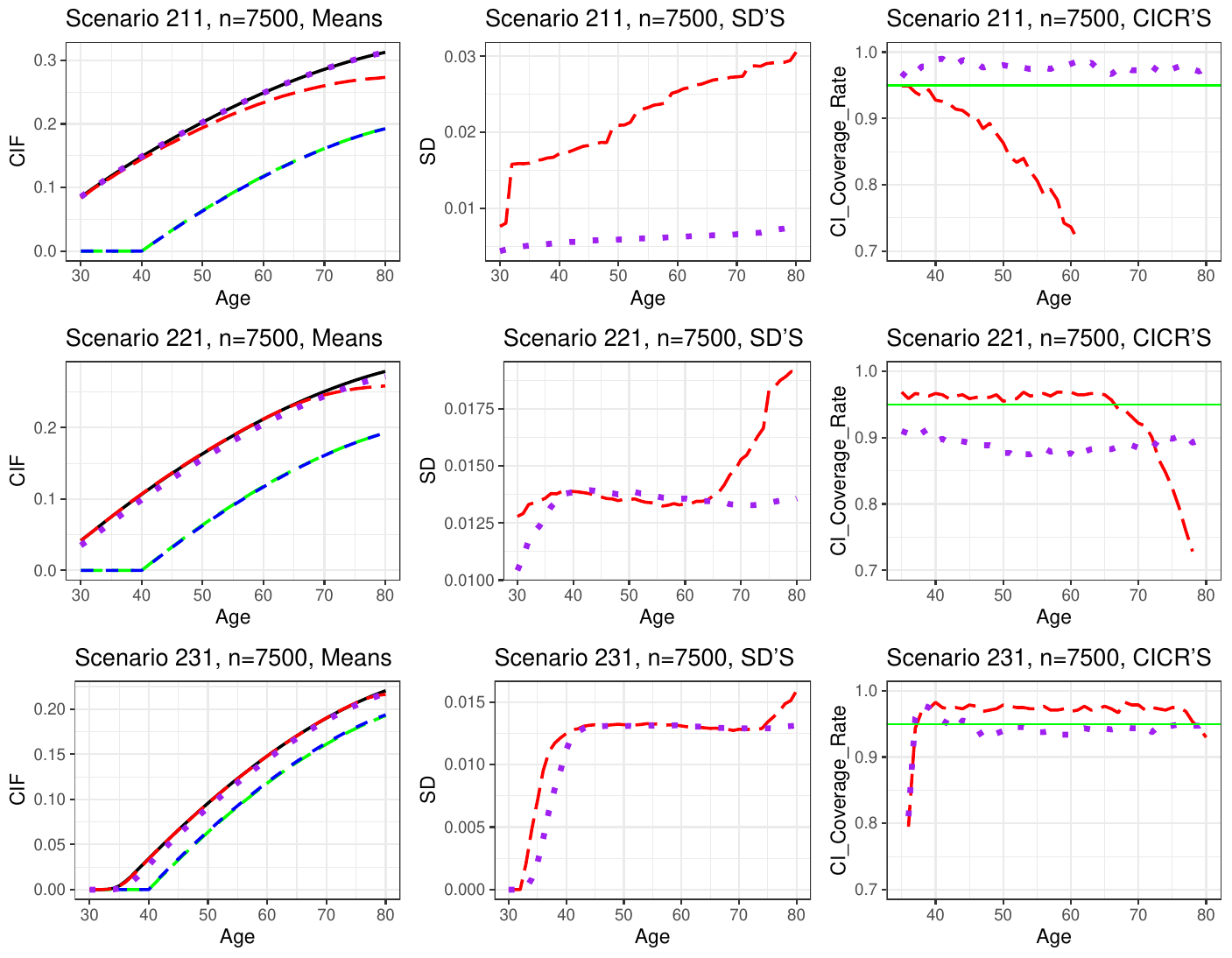}
	\caption{Simulation results for configurations 211, 221, and 231: Mean over estimates, standard deviations (SD), and 
		empirical coverage rates of 95\% pointwise confidence intervals for the Gorfine-Zucker-Shoham (GZS) and 
		and new estimators. Mean over estimates and target CIF curve shown also for the Aalen-Johansen (AJ) estimator.
		Solid black line – true curve, green long-dashed line - target of AJ estimator,
		blue dashed line – AJ estimator, red long-dashed line - GZS estimator, 
		purple dotted line – new estimator. }\label{fig:f8} 
\end{figure}

\begin{figure}[H]
	\centering
	\includegraphics[width=1\textwidth, height=0.85\textheight]{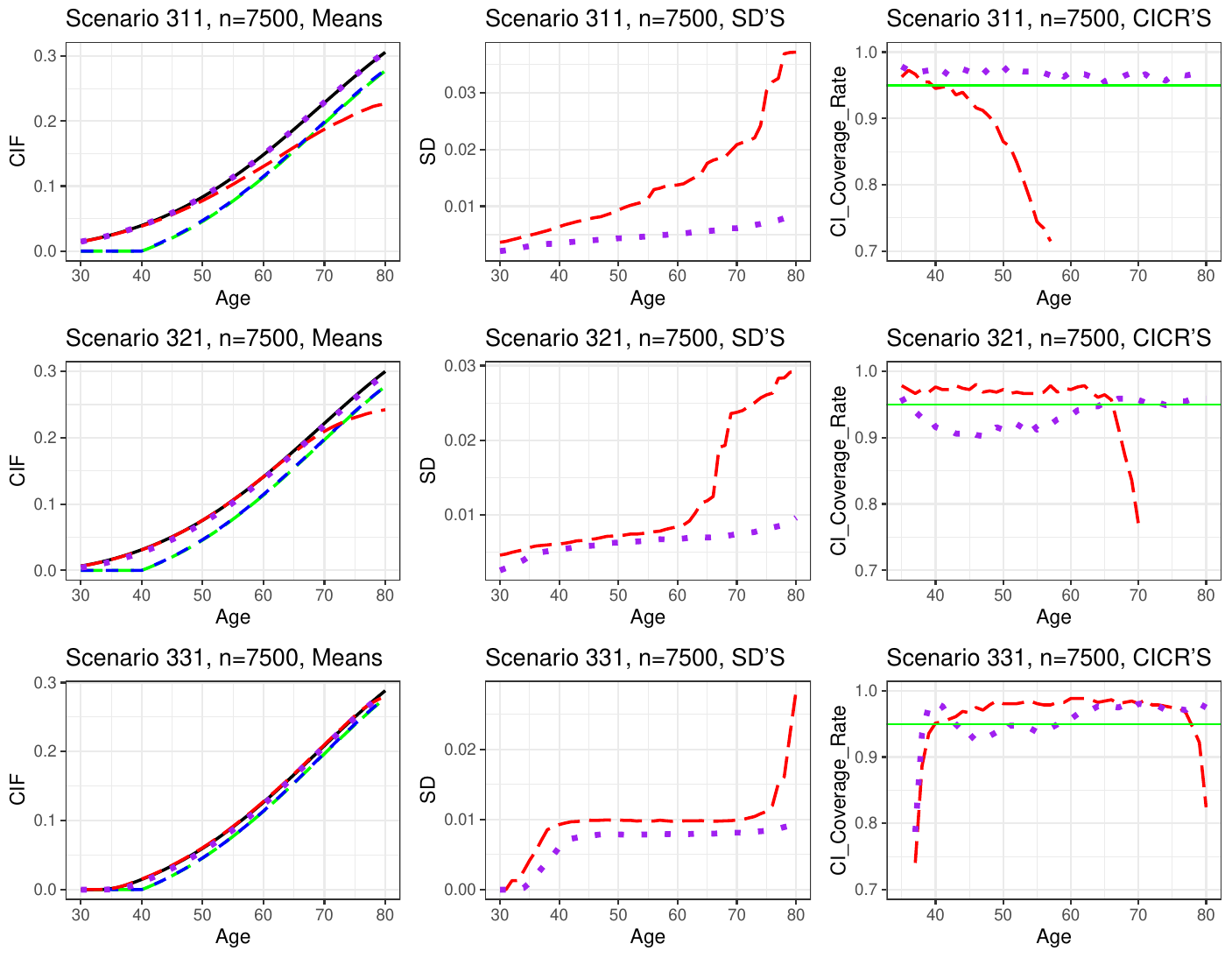}
	\caption{Simulation results for configurations 211, 221, and 231: Mean over estimates, standard deviations (SD), and 
		empirical coverage rates of 95\% pointwise confidence intervals for the Gorfine-Zucker-Shoham (GZS) and 
		and new estimators. Mean over estimates and target CIF curve shown also for the Aalen-Johansen (AJ) estimator.
		Solid black line – true curve, green long-dashed line - target of AJ estimator,
		blue dashed line – AJ estimator, red long-dashed line - GZS estimator,  
		purple dotted line – new estimator. }\label{fig:f9} 
\end{figure}

\begin{figure}[H]
	\centering
	\includegraphics[width=1\textwidth, height=0.85\textheight]{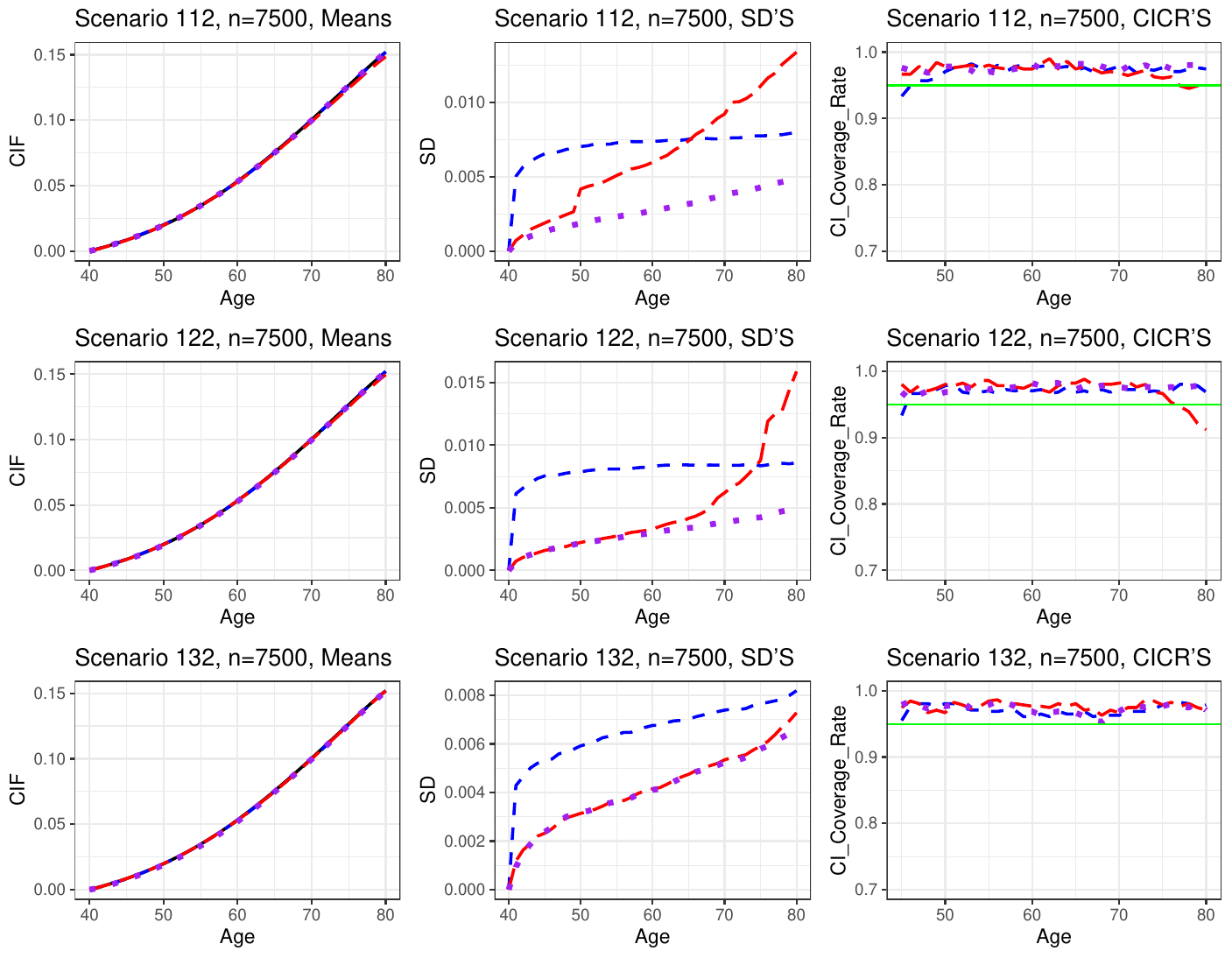}
	\caption{Simulation results for configurations 112, 122, and 132: Mean over estimates, standard deviations (SD), and 
		empirical coverage rates of 95\% pointwise confidence intervals, for the Aalen-Johansen (AJ), Gorfine-Zucker-Shoham (GZS),
		and new estimators. Solid black line – true curve, blue dashed line – AJ estimator, red long-dashed line - GZS estimator, 
		purple dotted line – new estimator. }\label{fig:f10} 
\end{figure}

\begin{figure}[H]
	\centering
	\includegraphics[width=1\textwidth, height=0.85\textheight]{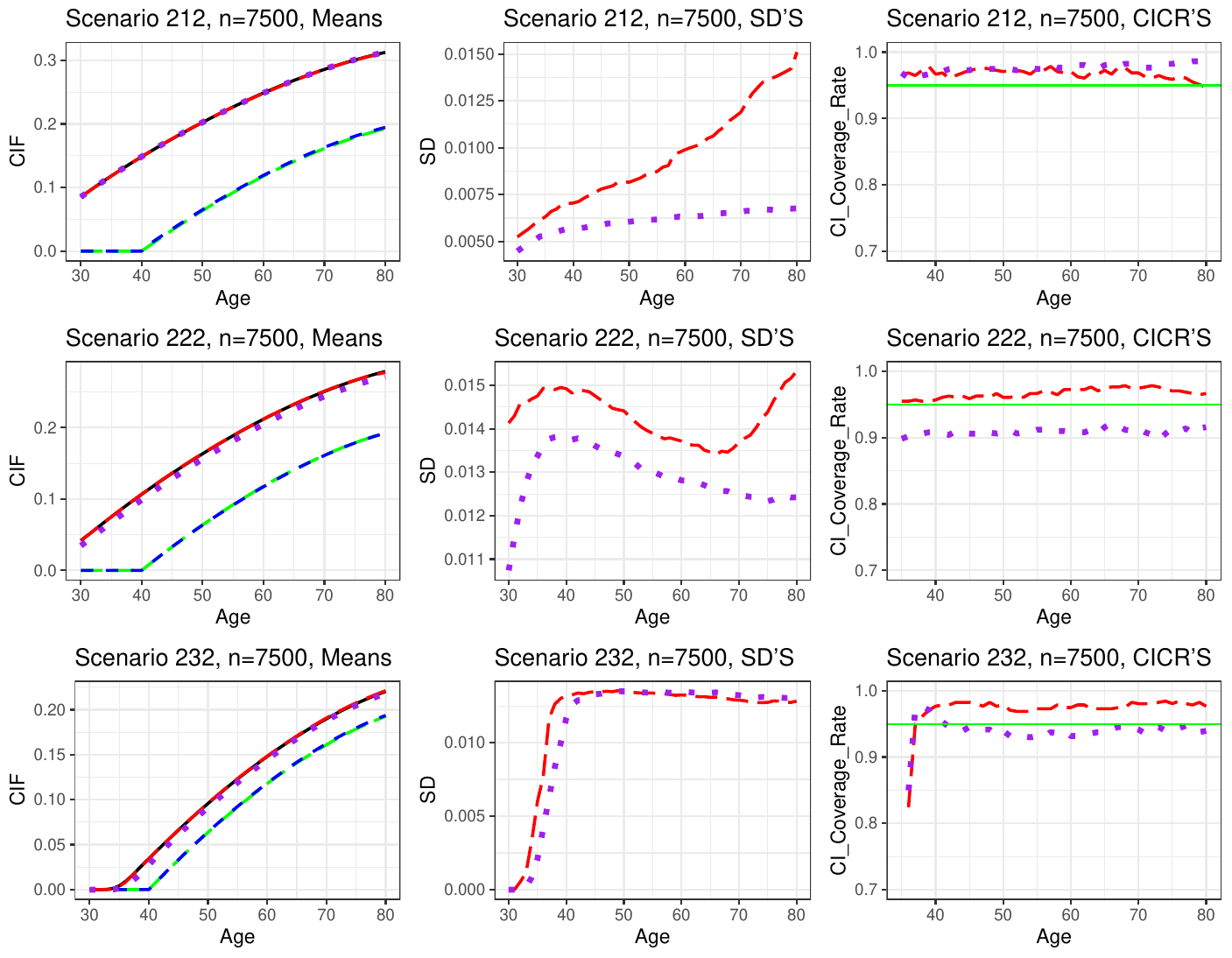}
	\caption{Simulation results for configurations 212, 222, and 232: Mean over estimates, standard deviations (SD), and 
		empirical coverage rates of 95\% pointwise confidence intervals, for the Aalen-Johansen (AJ), Gorfine-Zucker-Shoham (GZS),
		and new estimators. Solid black line – true curve, blue dashed line – AJ estimator, red long-dashed line - GZS estimator, 
		purple dotted line – new estimator. }\label{fig:f11} 
\end{figure}

\begin{figure}[H]
	\centering
	\includegraphics[width=1\textwidth, height=0.85\textheight]{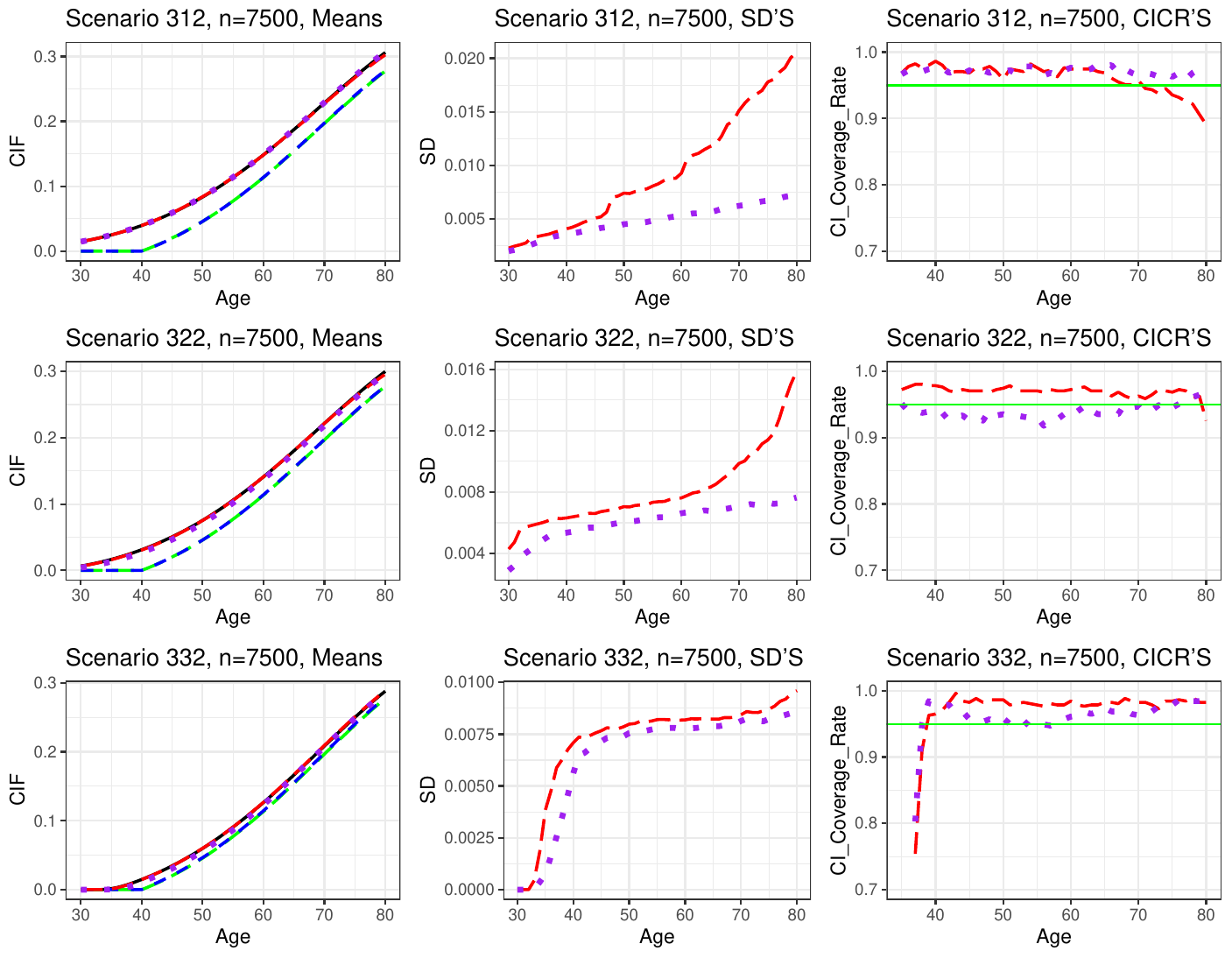}
	\caption{Simulation results for configurations 312, 322, and 332: Mean over estimates, standard deviations (SD), and 
		empirical coverage rates of 95\% pointwise confidence intervals, for the Aalen-Johansen (AJ), Gorfine-Zucker-Shoham (GZS),
		and new estimators. Solid black line – true curve, blue dashed line – AJ estimator, red long-dashed line - GZS estimator, 
		purple dotted line – new estimator. }\label{fig:f12} 
\end{figure}

\section{UK Biobank Female Data Analysis  - Additional Results }

Table \ref{tbl:UKB-not-interesting} and Figure \ref{fig:UKB-not-interesting} present results for four additional cancers: brain, esophagus, acute myeloid leukemia and lung. For these phenotypes, the CIF estimators from all three methods are highly similar, likely because the interval between diagnosis and death is shorter, so deaths after disease diagnosis are observed more completely.

\vspace*{1cm}

\renewcommand{\baselinestretch}{1}

\begin{table}[H]
	\centering
	\caption{Results of 95\% confidence band width (Band) and mean 95\% pointwise confidence interval width (Pointwise) for four phenotypes in female participants from the UK Biobank, obtained using the Aalen--Johansen estimator (AJ), the estimator of \cite{gzs2025} (GZS), and the proposed estimator (New)}\label{tbl:UKB-not-interesting}
	\begin{tabular}{|rcc|cc|cc|cc|}
		\hline
		& & & \multicolumn{2}{|c|}{AJ} & \multicolumn{2}{|c|}{GZS} & \multicolumn{2}{|c|}{New}  \\
		Cancer type & prevalent & incident & Band & Pointwise & Band & Pointwise & Band & Pointwise \\
		\hline
		Brain & 58 &  312 &  0.0008  &  0.0683 &  0.0008 &  0.0743 &   {\bf 0.0006} & {\bf 0.0003 } \\
		Esophagus & 34 & 278 &   0.0008 &  0.0444 & 0.0008 &  0.0039 & {\bf 0.0007} &  {\bf 0.0016}  \\
		AML & 47 & 135 & 0.0006 &  0.0179 & 0.0006 &  0.0088 & {\bf 0.0005} &  {\bf 0.0002}  \\
		Lung & 157 & 1910 & 0.0021 &  0.0599 & 0.0039 & 0.0071 &  {\bf 0.0018}&  {\bf 0.0019} \\
		\hline
	\end{tabular}  
\end{table}

\begin{figure}[H]
	\centering
	\includegraphics[width=0.9\textwidth, height=0.65\textheight]{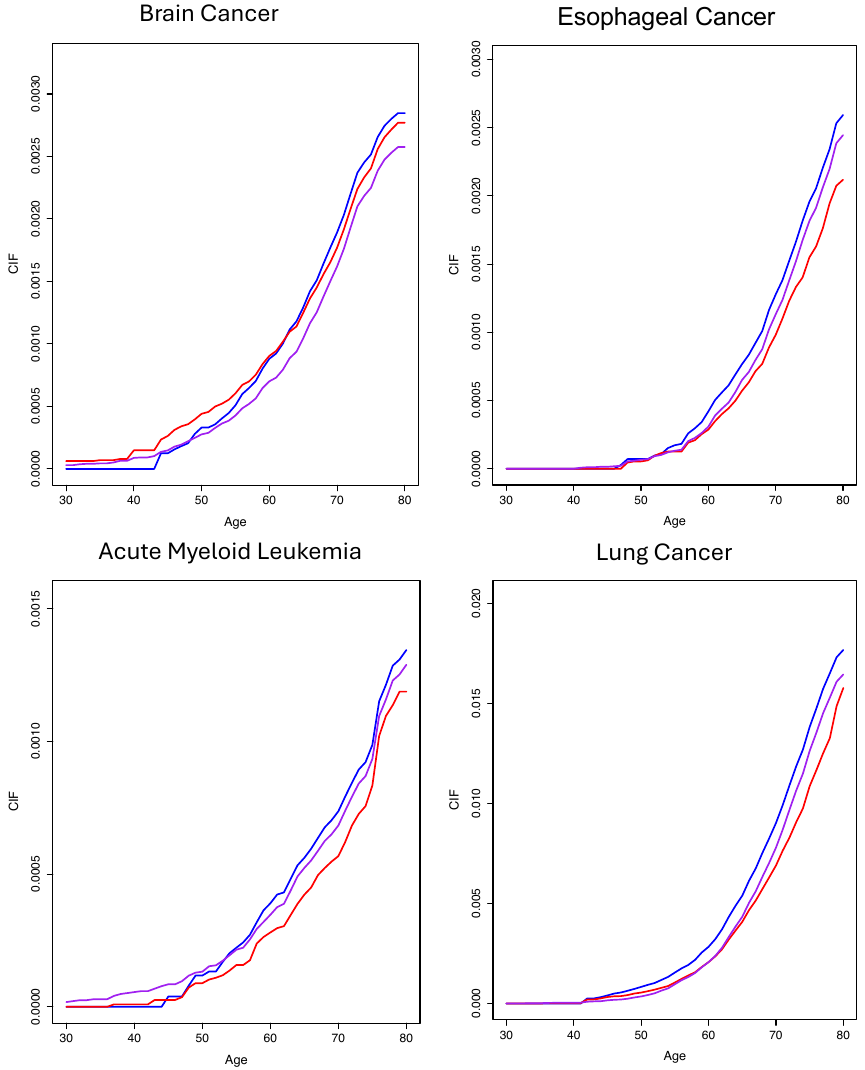}
	\caption{Estimated CIFs for four phenotypes in female participants from the UK Biobank, obtained using the Aalen--Johansen estimator (AJ, blue), the estimator of \cite{gzs2025} (GZS, red), and the proposed estimator (New, purple).}\label{fig:UKB-not-interesting}
\end{figure}

\clearpage

\end{document}